\documentclass[12pt]{amsart}

\usepackage{amssymb}
\usepackage{amsmath,amsthm,amsfonts}  
 \usepackage{graphicx}
\usepackage{pdfsync}
\usepackage{fancyhdr}
\usepackage{cancel}
\usepackage[applemac]{inputenc}
\usepackage{float}
\usepackage{mathtools}
\usepackage[default]{frcursive}
\usepackage[T1]{fontenc}
\usepackage[normalem]{ulem}
\usepackage{amsaddr}
\usepackage{color}
\usepackage{xcolor}
\usepackage[breaklinks,colorlinks]{hyperref}
\hypersetup{
    colorlinks, %set true if you want colored links
    linktoc=all, %set to all if you want both sections and subsections linked
    linkcolor=red,  %choose some color if you want links to stand out
  citecolor=hyptxt,
  urlcolor=blue}
  \hypersetup{
  citebordercolor=,
  filebordercolor=green,
  linkbordercolor=blue
}
\definecolor{hyptxt}{rgb}{0.7, 0.4, 0.9}
\usepackage[backend=biber,citestyle=numeric-comp, sorting=none,giveninits=true,hyperref=true,backref=true]{biblatex}

\usepackage[full]{textcomp} % to get the right copyright, etc.
     \usepackage{fbb} % osf for text, lining for math
     \usepackage[scaled=.95]{cabin}
     \usepackage[varqu,varl]{inconsolata}% typewriter
     \usepackage[libertine,bigdelims,vvarbb]{newtxmath}
     \usepackage[cal=boondoxo]{mathalfa}% less slanted than STIX
     \usepackage[T1]{fontenc}
     \usepackage{bbm}
\useosf
\newtheorem{defi}{Definition}[section]
\newtheorem{prop}{Proposition}[section]

\newtheorem{theo}{Theorem}[section]

\definecolor{hervecolor}{rgb}{0.8,0,0.7}

\newcommand{\ket}[1]{|\kern.3ex#1\kern.3ex\rangle}
\newcommand{\bra}[1]{\langle\kern.3ex #1 \kern.3ex|}
\newcommand{\scalar}[2]{\langle\kern.3ex #1 \kern.3ex|\kern.3ex#2\kern.3ex\rangle}

\newcommand{\ii}{\mathsf{i}}

\def\calH{{\mathcal H }}

\def\R{\mathbb{R}}

\def\lg{\langle }
\def\rg{\rangle }

\def\bu{\mathbbm{1}}

\def\ud{\mathrm{d}}

\def\sfN{\mathsf{N}}

\def\vv{\pmb{v}}
\def\vet{\pmb{\eta}}
\def\vph{\pmb{\phi}}

\def\vbe{\pmb{\beta}}
\def\vze{\pmb{\zeta}}
\def\sfD{\mathsf{D}}

\def\vt{\pmb{t}}
\def\vI{\pmb{I}}
\def\vJ{\pmb{J}}
\def\vK{\pmb{K}}
\def\bi{\widehat{\boldsymbol{\imath}}}
\def\bj{\widehat{\boldsymbol{\jmath}}}
\def\bk{\widehat{\boldsymbol{k}}}
\numberwithin{equation}{section}

\usepackage[low-sup]{subdepth}
\usepackage[top=2cm,bottom=2cm,left=2cm,right=2cm]{geometry}

\bibliography{Real_POVMs_Plane}

\begin{document}
\date{\today}
 
\title[Euclidean POVMs]{Quantum formalism on the plane: POVM-Toeplitz quantization, Naimark theorem and linear polarisation of the light}
\author[R. Beneduci, E. Frion, J.-P. Gazeau \& A. Perri]{\small 
Roberto Beneduci$^{\mathrm{a}}$, Emmanuel Frion$^{\mathrm{b}}$, Jean-Pierre Gazeau$^{\mathrm{c}}$ and Amedeo Perri$^{\mathrm{d}}$}

\address{\parbox{\linewidth}{\textit{$^{\mathrm{A}}$ Dipartimento di Fisica, Universit\`a della Calabria}\\
		\textit{and Istituto Nationale di Fisica Nucleare, Gruppo c. Cosenza,} \\ 
		\textit{87036 Arcavacata di Rende (Cs), Italy}}} 

\address{\parbox{\linewidth}{\textit{$^{\mathrm{B}}$ Helsinki Institute of Physics, P.~O.~Box  64, \\ FIN-00014 University of Helsinki, Finland }}} 

\address{\parbox{\linewidth}{\textit{$^{\mathrm{C}}$ Universit\'e Paris Cit\'e,}\\ \textit{CNRS, Astroparticule et Cosmologie}  \\
		\textit{75013 Paris, France}}\\}
	
\address{\parbox{\linewidth}{\textit{$^{\mathrm{D}}$ Independent Scholar,}\\ \textit{Via Stazione 4, Spezzano della Sila}  \\
		\textit{87058 Cosenza, Italy}}\\}

\email{ \href{mailto:roberto.beneduci@unical.it}{roberto.beneduci@unical.it}, \href{mailto:emmanuel.frion@helsinki.fi}{emmanuel.frion@helsinki.fi}, \href{mailto:gazeau@apc.in2p3.fr}{gazeau@apc.in2p3.fr}}

%\email{e-mail:  \href{mailto:roberto.beneduci@unical.it}{roberto.beneduci@unical.it}, \; \href{mailto:gazeau@apc.in2p3.fr}{gazeau@apc.in2p3.fr}, \;
% \href{mailto:emmanuel.frion@cosmo-ufes.org}{emmanuel.frion@cosmo-ufes.org}}

{\abstract{\vspace{1cm} We investigate two aspects of the elementary example of POVMs on the Euclidean plane, namely their status as quantum observables and their role as quantizers in the integral quantization procedure. The compatibility of POVMs in the ensuing quantum formalism is discussed, and a Naimark dilation is found for the quantum operators. The relation with Toeplitz quantization is explained.  A physical situation is discussed, where we describe the linear polarization of the light with the use of Stokes parameters. In particular, the case of sequential measurements in a real bidimensional Hilbert space is addressed. An interpretation of the Stokes parameters in the framework of unsharp or fuzzy observables is given. Finally, a necessary condition  for the compatibility of two dichotomic fuzzy observables which provides a condition for the approximate joint measurement of two incompatible sharp observables is found.}}
		
%We then consider two symmetric POVMs, as well as two real vectors $v_1(r_1,\phi_1)$ and $v_2(r_2,\phi_2)$ obeying a set of conditions constrained by the operators eigenvalues. We subsequently use this set of conditions to show that the compatibility condition between POVMs reads $||v_1-v_2||+||v_1+v_2||<2$. We finally generalize some of our  results to the real $n$-dimensional case.}}

\maketitle

\clearpage

\tableofcontents

\clearpage

\section{Introduction}
\label{intro}

Positive operator-valued measures (POVMs) play a central role in the Hilbert space formulation of quantum mechanics. Indeed, observables are represented by positive operator-valued measures, states are represented by positive trace-class operators with trace one and observables and states are in statistical duality (see next section). One of the main features of quantum physics is the existence of incompatible observables that cannot be measured simultaneously \cite{busch2016quantum}. In other words, measurement of two observables in quantum physics can show some incompatibility. This reflects in purely quantum phenomena, such as the violation of Bell inequalities \cite{bell1964einstein,bell1966problem} or the Kochen-Specker contextuality \cite{Koc}, with implications in quantum information processing \cite{spekkens2009preparation, kleinmann2011memory, grudka2014quantifying, chailloux2016optimal, abramsky2017contextual, schmid2018contextual, duarte2018resource, ghorai2018optimal} and in explaining the power of quantum computation \cite{galvao2005discrete,cormick2006classicality,anders2009computational,howard2014contextuality,hoban2014measurement,karanjai2018contextuality,frembs2018contextuality,raussendorf2019phase}.

Spectral measures, which are in one-to-one correspondence with self-adjoint operators, are particular cases of POVMs. The compatibility of two observables represented by spectral measures is equivalent to their commutativity. However, that is not true for general (unsharp) observables described by POVMs for which compatibility does not imply commutativity. This opens new perspectives and raises the question of compatibility between POVMs. General measurements  can  realize  arbitrary  joint  measurability  structure \cite{kunjwal2014quantum,andrejic2020joint}, and studying incompatible sets can prove relevant for quantum information protocols \cite{skrzypczyk2019all,pusey2015verifying,buscemi2019complete}. 

Several works have been devoted to find conditions for the compatibility of POVMs. See \cite{busch2016quantum} for a review, see also \cite{Wolf}, where the necessary and sufficient condition for compatibility of two dichotomic observables is reformulated as semidefinite program, and  \cite{beneduci2017joint} where it is shown that two POVMs are jointly measurable if and only if they can be dilated to two commuting orthogonal POVMs, though this result cannot be extended to the case of more than two POVMs \cite{heunen2014quantum}. In the present paper, we analyze the simple case of two dichotomic POVMs in a two-dimensional real Hilbert space, i.e. the Euclidean plane, and show that a condition for the compatibility can be derived starting from its definition. We provide an interpretation of fuzzy dichotomic observables in terms of the Stokes parameters for the polarization of an electromagnetic wave; this suggests an interpretation of the approximate joint measurement of incompatible observable as well. From the mathematical viewpoint, this is a particular case of the results found in Ref. \cite{busch2016quantum}, chapter 14, in which the complex two-dimensional case is analyzed; it is the real Hilbert space case obtained when the coefficient multiplying the Pauli matrix $\sigma_2$ in the definition of the self-adjoint operators is equal to zero, \textit{i.e.}, we have $A=a_0\bu+\mathbf{n}\cdot\mathbf{\sigma}$ with $\mathbf{n}=(n_1,0,n_3)$, $\mathbf{\sigma}=(\sigma_1,\sigma_2,\sigma_3)$ and $\bu$ the identity operator. In the framework of electromagnetic wave polarization, that means we are focusing on the linear polarization and are not taking into consideration the circular polarization.  

Besides their role as quantum observables, POVMs also provide a quantization procedure based on integral calculus. The properties of this integral quantization procedure were studied in \cite{bergeron2014quantizations}, where it was shown that a choice of POVM gives a well-defined quantum model for a system, and offers the possibility of a regularization effect as well as a full probabilistic content. These features have been leveraged in quantum cosmological models, where the Big Bang singularity is replaced by a bounce \cite{bergeron2014smooth,bergeron2015smooth,bergeron2015smoothquantum,bergeron2015singularity,bergeron2016nonadiabatic,bergeron2016vibronic,bergeron2017spectral,bergeron2018primordial,bergeron2018integrable,bergeron2020quantum,almeida2018three,frion2019affine}. POVMs play a key role in the phase-space formulation of quantum mechanics as well \cite{Ali,Prugovecki,Schroeck}.

Finally, we recall that an interesting class of measurements represented by POVMs concerns weak measurements. They are measurements which disturb (in the average) a quantum system very little, while extracting a little bit of information (see \textit{e.g.} \cite{brun2002simple,tamir2013introduction}). Weak measurements can be used experimentally to show the breaking of Bell's inequalities \cite{palacios2010experimental}, which can be used in order to set up experiments preserving entanglement even when performing many sequential measurements \cite{kim2012protecting,foletto2020experimental}. In the same vein, weak measurements are also used to test the noise-disturbance relation stemming from the Heisenberg uncertainty principle \cite{ozawa2003universally} in various experiments \cite{erhart2012experimental,kaneda2014experimental,inoue2020violation}.

The content of the present article shows various non-trivial aspects connected with the role of POVMs as both observables and quantization maps. One of its objective is to describe a sequential measurement with  POVMs acting on the Euclidean plane, certainly one of the simplest examples of Hilbert space. We particularly focus  on the question of incompatibility.  In this line, we begin by reviewing the definition of a POVM and compatibility between POVMs in Section \ref{genset}.  In Section \ref{integral quantization}, we give the definition and main properties of integral quantization and its concomitant semi-classical portraits by insisting on the  genuine role of  POVMs  in implementing the procedure. In Section \ref{toepnaim} we establish an original result, namely the relationship between our approach and  the so-called Toeplitz quantization, which allows us to find the Naimark dilation of our quantum operators. In Section \ref{intquant}, we apply our formalism and Euclidean plane to the circle. We show that the integral POVM quantisation of a real function defined on the circle  results in a $2\times 2$ matrix acting on the Euclidean plane. We then describe the corresponding semi-classical portrait or symbol calculus. We end the section by illustrating the Naimark-Toeplitz correspondence within this $2d$ Euclidean framework.   A nice  physical application of our approach is developed  in Section \ref{quantmeas}, where  we establish the link with Stokes parameters for the linear polarization of electromagnetic waves and describe the interaction of a linearly polarized light with a polarizer as an example of quantum measurement. We then discuss the approximate measurement of two incompatible sharp polarization observables $E$ and $E'$ realized by the joint measurement of their fuzzy versions $F$ and $F'$. This requires a characterization of the compatibility of two POVMs. We begin Section \ref{Compatibility} with an illustration of the incompatibility problem, then we derive the necessary conditions for the compatibility of dichotomic observables. As a last result, we show there is a direct relation between the fuzziness of a measurement and the degree of linear polarization of a physical observable}. We generalize some technical aspects of our results to $\mathbb{R}^n$ in Appendix \ref{QSON}.

\section{General setting}
\label{genset}

What is an observable? In standard quantum mechanics an observable is represented by a self-adjoint operator but it has been argued \cite{ludwig2012foundations,kraus1983states,holevo2011probabilistic,ali1982geometrical,busch2007approximate,busch2016quantum,schroeck2013quantum} that the choice of self-adjoint operators as representative of observables is too restrictive. It has been suggested that an observable is better represented by a POVM of which the spectral measures are particular examples. Before we give motivation for such a generalization of the concept of observable, let us recall the main definitions and properties of POVMs \cite{beneduci2017joint}. In the following, the set of positive trace-class self-adjoint operators with trace one is denoted by $S(\mathcal{H})$. 

\subsection{Definition of a POVM}
\label{defpovm}
\begin{defi}
A normalized Positive-Operator Valued measure (POVM) is a map $F:\mathcal{B}(\Omega)\to\mathcal{L}^+_s(\mathcal{H})$ from the Borel $\sigma$-algebra of a topological space $\Omega$ to the space of linear positive self-adjoint operators such that:
    \begin{align}
    \label{unF}
    F\left(\bigcup_{n=1}^{\infty}\Delta_n\right)&=\sum_{n=1}^{\infty}F(\Delta_n) \;,\\ \label{unFres}
    F(\Omega)&=\bu \;,
    \end{align}
    \noindent 
 where $\{\Delta_n\}$ is a countable family of disjoint sets in $\mathcal{B}(\Omega)$ and the series converges in the weak operator topology. The POVM is said to be real if $\Omega=\mathbb{R}$. A projection-valued measure (PVM) is a POVM such that $F(\Delta)$ is a projection operator for every $\Delta\in\mathcal{B}(\Omega)$.
 \label{definition povm}
\end{defi}

\noindent
Spectral measures are real PVMs. By the spectral theorem, spectral measures are in a one-to-one correspondence with self-adjoint operators. Therefore, if we use POVMs to represent observables, the set of observables represented by self-adjoint operators is just a subset of the set of observables.  

The motivation for such an extension is rooted in the statistical nature of the measurement process.  It is well known that by repeating the measurement $M$ of a system $S$ in the state $\rho\in S(\mathcal{H})$ we, in general, obtain different results (that we suppose contained in a topological space $\Omega$). On the other hand, by repeating the measurement $M$ several times, we observe a statistical regularity, which is encoded in the statistical distribution $f^M_\rho$ of the measurement outcomes, and depends only on the state of the system. In other words, we have a map $\rho\mapsto f^M_\rho$ ensuring the existence of regularities of a statistical nature. Note that this is a necessary condition for the existence of the scientific enterprise. In its turn, the statistical distribution $f_\rho^M$ is represented by a probability measure $\mu^M_\rho:\mathcal{B}(\Omega)\to[0,1]$. At this point of our argument, it is better to specify that the state of the system is defined by the procedure we followed to prepare the system before the measurement, and to recall that the space $S(\mathcal{H})$ of the states of a given system is a convex space (this corresponds to the existence of mixed states). For example, a filtering operation ensures that if a beam of electrons passes through the filter, the electrons in the beam must immediately afterward have positions within a restricted range of values.

The previous analysis can be summarized by saying that the measurement process $M$ defines a map from \lq\lq{}states\rq{}\rq{} to \lq\lq{}probability measures,\rq{}\rq{} $\rho\mapsto\mu^M_\rho$. Moreover, if the state is given by a convex combination of states, $\rho=t\rho_1+(1-t)\rho_2$, the measurement outcomes are given by a convex combination of probability measures, $t\mu^M_1+(1-t)\mu^M_2$. In mathematical terms, the map $\rho\mapsto \mu^M_\rho$ is affine. Note that the previous argument is of a general nature. It applies both to classical and quantum formalism (in mechanics, optics, ...). In the classical case, for example, the states are represented by probability measures on the phase space $\Gamma$, and the measurements are affine maps that transform probability measures on $\Gamma$ into probability measures on the space of possible outcomes of the measurement process, $\Omega$  \cite{holevo2011probabilistic}. Actually, the scheme we just described applies to any physical model where the states of the system define a convex set (see \cite{busch2016quantum} chapter 23). In order to introduce the concept of observable in this scheme, it is sufficient to note that different measurements can correspond to the same affine map. In this case, we say that the measurements are equivalent and divide the set of measurement into equivalence classes. An observable is an equivalent class of measurements. The same can happen with states: two states $\rho_1$ and $\rho_2$ are said to be equivalent if they cannot be separated by the set of measurements, \textit{i.e.}, if $\mu_1^M=\mu_2^M$ for every measurement $M$.   

Now, let us focus on the particular case of the quantum framework where, to each system, a complex and separable Hilbert space $\mathcal{H}$ is associated  and the states are represented by density operators, \textit{i.e.} positive, bounded self-adjoint operators with trace $1$. We have the following characterization of the measurement process:

\begin{theo}[Holevo, \cite{holevo2011probabilistic}]
There is a one-to-one correspondence between POVMs $F:\mathcal{B}(\Omega)\to\mathcal{L_s^+(H)}$ and affine maps $S(\mathcal{H})\mapsto\mathcal{M}_+(\Omega)$ from states to probability measures which is given by $\mu(\Delta)=Tr(\rho F(\Delta))$, $\Delta\in\mathcal{B}(\Omega)$, $\rho\in S(\mathcal{H})$, $\mu\in\mathcal{M}_+(\Omega)$.
\end{theo}

\noindent
%The proof of the theorem does not depend on the Hilbert space to be complex so that it is still true in the case of a real Hilbert space. In particular it is true in the case $\mathcal{H}=\mathbb{R}^n$.  

Note that equivalent measurements are represented by the same POVM and this establishes a one-to-one map between observables and POVMs. Moreover, $\mu_\rho^M=Tr[F(\Delta)\rho]$ establishes a statistical duality between observables (POVMs) and states (density operators) \cite{busch2016quantum}.

\subsection{Naimark theorem and compatibility of POVMs}
\label{naimark compatibility}

In the previous section, we have seen that observables are represented by POVMs and that observables represented by spectral measures (self-adjoint operators) define a subset of the space of observables. Now,  we recall Naimark\rq{}s dilation theorem which ensures that every POVM is the projection of a projector-valued measure defined in an extended Hilbert space. 

\begin{theo}[Naimark, \cite{gelfand1943imbedding}]
Let $F:\mathcal{B}(\Omega)\to\mathcal{L}^+_s(\mathcal{H})$ be a POVM in a Hilbert space $\mathcal{H}$. Then, there is an extended Hilbert space $\mathcal{H}^+$ and a projector-valued (PV) measure $E^+:\mathcal{B}(\Omega)\to\mathcal{L}_s(\mathcal{H}^+)$ such that 
$$PE^+(\Delta)\psi=F(\Delta)\psi,\quad \psi\in\mathcal{H}, \quad \Delta\in\mathcal{B}(\Omega) \,,$$
where $P$ is the projection operator onto $\mathcal{H}$.
\end{theo}

\noindent
Naimark\rq{}s theorem provides a necessary and sufficient condition for the compatibility of two POVMs.

\begin{defi}
Two POVMs $F_1:\mathcal{B}(\Omega_1)\to\mathcal{L}^+_s(\mathcal{H})$ and $F_2:\mathcal{B}(\Omega_2)\to\mathcal{L}^+_s(\mathcal{H})$ are compatible if there is a third POVM $F:\mathcal{B}(\Omega_1\times\Omega_2)\to\mathcal{L}^+_s(\mathcal{H})$ of which $F_1$ and $F_2$ are the marginals, \textit{i.e.},
$$F_1(\Delta_1)=F(\Delta_1\times\Omega_2),\quad F_2(\Delta_2)=F(\Omega_1\times\Delta_2).$$
\end{defi}

\noindent
We recall that the symbol $\mathcal{B}(\Omega_1\times \Omega_2)$ denotes the product $\sigma$-algebra generated by the family of sets $\{\Delta_1\times\Delta_2\,\,:\,\,\Delta_1\in\mathcal{B}(\Omega_1),\,\Delta_2\in\mathcal{B}(\Omega_2)\}$.

\begin{theo}[\cite{beneduci2017joint}]\label{cnsNaimark}
Two POVMs $F_1:\mathcal{B}(\Omega_1)\to\mathcal{L}^+_s(\mathcal{H})$ and $F_2:\mathcal{B}(\Omega_2)\to\mathcal{L}^+_s(\mathcal{H})$ are compatible if and only if there are two Naimark extensions $E^+_1:\mathcal{B}(\Omega_1)\to\mathcal{L}_s(\mathcal{H}^+)$ and $E^+_2:\mathcal{B}(\Omega_2)\to\mathcal{L}_s(\mathcal{H}^+)$ such that $[E^+_1,E^+_2]=0$. 
\end{theo}

\section{An outline of integral quantization and semi-classical portraits}
\label{integral quantization}
In this section, we recall  the fundamental role played by POVMs as quantizers of classical models of physical systems, \textit{i.e.} in the building of self-adjoint operators and their subsequent PV measures \cite{bergeron2014quantizations,gazeau2015positive}.
\subsection{Quantization map}

%A quantum model for a system should a classical counterpart from which it is built through some quantization procedure.  
% Then we explain how the classical counterparts of the above $2\times2$ symmetric matrices, viewed as quantum observables, are the quantized versions of real-valued functions on the circle, and that there is a one-to-one  relation if we restrict the functions to be linear in  elements in the finite set $\{1, \cos{2\phi}, \sin{2\phi}\}$. 
To define a quantization procedure for functions defined on a set $X$ (\textit{e.g.} a phase space), the map from elements in a  ``classical'' vector space to operators in some Hilbert space must fulfill basic requirements. A minimal program should meet the four following criteria \cite{bergerongazeau2014,ali2014coherent}:

\begin{enumerate}
	\item \textit{Linearity}. Quantization is a linear map $f\mapsto A_f$:
	\begin{align} \label{qmap1}
	\mathfrak{Q}:\mathcal{C}(X)\mapsto\mathcal{A}(\calH)\, , \qquad \mathfrak{Q}(f) = A_{f}\, , 
	\end{align} 
	where
	\begin{itemize}
		\item $\mathcal{C}(X)$ is a vector space of complex or real-valued functions $f(x)$
		on a set $X$, \textit{i.e.} a ``classical'' mathematical model, 
		\item $\mathcal{A}(\calH)$ is a vector space of linear operators  in some real or complex Hilbert space $\calH$, \textit{i.e.}, a ``quantum'' mathematical model, notwithstanding the question of common domains in the case of unbounded operators.  
	\end{itemize}
	\item \textit{Unity}. The map \eqref{qmap1} is such that the function $f=1$ is mapped to the identity operator $\bu$ on $\calH$.
	\item \textit{Reality}. A real function $f$ is mapped to a  self-adjoint operator $A_{f}$ in $\calH$ or, at least, a symmetric operator (in the infinite-dimensional case).
	\item \textit{Covariance}. Defining the action of a symmetry group G on X by $\left(g,x\right)\in G\times X$ such as $\left(g,x\right) \mapsto g\cdot x \in X$, there is a unitary representation $U$ of $G$ such that  $A_{T(g)f}= U(g)A_f U(g^{-1})$, with $(T(g)f)(x)= f\left(g^{-1}\cdot x\right)$.
\end{enumerate}

In the remainder of this paper, we suppose $\calH$ be a finite Hilbert space. Therefore, the quantized objects $A_f$ are matrices \cite{gazeau2009coherent}. Note that the previous requirements are minimal. In this work, we make use of POVMs to define the quantization procedure, and it is therefore necessary to consider  additional minimal structures, namely a measure $\nu$ on $X$, the $\sigma$-algebra defined in \eqref{definition povm}, and local compactness. POVMs can be used in conjunction with integral quantization to fulfill  the first three above requirements. However, covariance needs some additional structure \cite{bergeron2019orientations,gazeau2015positive}. In particular, Euclidean POVMs provide covariance, as we will see in section \ref{sec:integquant}.

The integral quantization of a function $f(x)$ on the measure space $(X,\nu)$ is defined by the linear map 
\begin{equation}
\label{iqmap}
	f\mapsto A_{f}=\int_{X}\,{\sf M}(x)\, f(x)\,\ud\nu(x)\,,
\end{equation}
where the family of  matrices ${\sf M}(x)$  solves the identity  as
\begin{equation}
\label{resunitM}
	X\ni x\mapsto\mathsf{M}(x)\,,\quad\int_{X}\,{\sf M}(x)\,\ud\nu(x)=\bu\,.
\end{equation} 
% Such structures appear in a natural way in the examples that we consider in this paper, namely the circle, $X= \mathbb{S}^1$, which is considered in the sequel as providing the simplest illustration of the method, and the group manifold SO$(n)$.  %and the group manifolds SO$(3)$ and more generally SO$(n)$. 

If the matrices ${\sf M}(x)$ are non-negative, the quantum operator related to the characteristic function on $\Delta$, $A(\chi_\Delta)$, defines a POVM . Indeed, the restriction to $\Delta$ of the quantization map \eqref{iqmap} 
\begin{align}
	F(\Delta):=A(\chi_\Delta)=\int_{X}\,{\sf M}(x)\, \chi_\Delta(x)\,\ud\nu(x)=\int_{\Delta}\,{\sf M}(x)\, \ud\nu(x) \;,
\end{align}
solves the identity \eqref{unFres} by definition. Therefore, we have shown two key roles of POVMs; they are the mathematical representatives of observables and they provide a quantization procedure.

\subsection{Semi-classical portraits}
\label{semclasspor}
A quantum model for a system must have a classical or semi-classical counterpart. We say there exists a \textit{dequantization} map acting on an operator $A_f$ to give back the original $f(x)$, in general with some corrections. This procedure gives a \textit{semi-classical portrait} of the original function $f$, often denoted $\check f(x)$ (Husimi and Wigner functions \cite{zachos2005quantum},  lower (Lieb)  \cite{lieb1973classical} or covariant (Berezin) \cite{berezin1975general} symbols). Given two families ${\sf M}_{\mathrm{a}}(x)$ (for ``analysis'') and ${\sf M}_{\mathrm{r}}(x)$ (for ``reconstruction") resolving the identity in the sense of \eqref{resunitM},
we here generalize  the definition of these lower symbols in the following way:
\begin{equation}
 \label{lowsymbgen}
A_{f}\mapsto\check{f}(x):=\mathrm{tr}(\mathsf{M}_{\mathrm{r}}(x)\, A_{f})= \int_{X}\,\mathrm{tr}(\mathsf{M}_{\mathrm{r}}(x){\sf M}_{\mathrm{a}}(x^{\prime}))\, f(x^{\prime})\,\ud\nu(x^{\prime})\,,
\end{equation} 
Choosing $\mathsf{M}_{\mathrm{a}}(x)=\rho_{\mathrm{a}}(x)$ and $\mathsf{M}_{\mathrm{r}}(x)=\rho_{\mathrm{r}}(x)$ to be  density matrices, the map $x^{\prime} \mapsto \mathrm{tr}(\rho_{\mathrm{r}}(x)\rho_{\mathrm{a}}(x^{\prime}))$ defines a probability
distribution $\mathrm{tr}(\rho_{\mathrm{r}}(x)\rho_{\mathrm{a}}(x^{\prime}))$ on the measure space $(X,\ud\nu(x^{\prime}))$, as it is easily proved by multiplying the resolution of the unity \eqref{resunitM} with $\rho_{\mathrm{r}}(x)$ and tracing the result. Thus, the expectation value of the operator $A_f$, $\check{f}(x)$, is built from the analysis POVMs based on the $\mathsf{M}_{\mathrm{a}}(x)$'s and given by:
\begin{equation}
 \label{locaver}
f(x)\mapsto\check{f}(x)=\int_{X}f(x^{\prime})\,\mathrm{tr}(\rho_{\mathrm{r}}(x)\rho_{\mathrm{a}}(x^{\prime}))\,\mathrm{d}\nu(x^{\prime})\,.
\end{equation}
%Equivalently, we can say that the semi-classical map of $f$ can be viewed as an unsharp measurement.
The map \eqref{locaver} represents in general a regularization of the original, possibly extremely singular, $f$. Standard cases are for $\rho_{\mathrm{a}}= \rho_{\mathrm{r}}$, in particular for rank-one density matrices (coherent states).

\subsection{POVMs and quantum measurement}

We base the content of this section on \cite{gazeau2015positive}, in which  the analogy of the present formalism with quantum measurement was discussed. 
In quantum mechanics, a physically relevant $A_f$ is a self-adjoint operator whose expectation value is the ``unsharp'' representation \cite{holevo2011probabilistic,busch2016quantum} 
\begin{align}
	\label{measexpect}
	\mathrm{tr}\left(\rho_m A_f\right) =  \int_{X}f(x)\,\mathrm{tr}(\rho_m\rho(x))\,\mathrm{d}\nu(x)\, .  
\end{align}
In compliance with the presentation given in Section \ref{defpovm},  the operator $$\rho_m =\sum_i q_i \ket{\phi_i} \bra{\phi_i}$$ is here the density matrix describing the mixing of the pure states $\ket{\phi_i}$, each associated with the outcome probability $q_i$. On the other hand, the expectation value \eqref{measexpect} of any quantum observable stems from the probability distribution $\mathrm{tr}\left(\rho_m \rho(x)\right)$. We note that, unlike in \eqref{locaver}, $\rho_m$ does not depend on any variable. Since the family of density operators used to construct our quantization procedure can provide the choice $\rho_m = \rho_r(x_0)$ for some $x_0$, we obtain a quantum measurement based on the pair of POVMs $\left(\rho_m, \rho(x)\right)$, which is equivalent to the pair $\left(\rho_r,\rho_a\right)$ introduced in \eqref{locaver}. We also note that the unsharp representation  can be associated with the ``sharp'' representation
\begin{equation}
\label{measspecdec}
\mathrm{tr}\left(\rho_m A_f\right) =  \int_{\R}\lambda\,\mathrm{tr}(\rho_m\,\mathrm{d}E_f (\lambda))\; ,
\end{equation}
coming from the self-adjoint operator $A_f$ spectral decomposition
\begin{equation}
\label{spmeas}
A_f = \int_{\R} \lambda\, \ud E_f(\lambda)\; , 
\end{equation}
with spectral measure $\mathrm{d}E_f(\lambda)$. While the sharp representation is directly associated to the PV spectral measure $E_f$, a family $\rho(x)$ resolving the identity is needed in the unsharp representation. However, classical and quantum spectra can be considered identical, regardless of the differences between their PV measure and POVMs. Finally, we observe that both the quantum operator $A_f$ and its semi-classical counterpart \eqref{locaver} involve probability distributions  with respect to some  POVMs. Therefore, quantifying an object and measuring it are done through a sole procedure.

\section{Integral quantization, Toeplitz quantization and Naimark dilation}
\label{toepnaim}
The so-called Berezin-Klauder-Toeplitz quantization (see \cite{landsman2006} for a clear presentation of it \footnote{We stress that the quantization procedure presented in \cite{landsman2006} is based on group action while our presentation does not.}), and in particular coherent state (CS) quantization, is one kind of integral quantization. In this section, we describe their relationship, which will allow us to find the Naimark dilation of our quantum operators. 

The   construction of  family of operators  solving the identity in \eqref{resunitM} in the CS case is implemented along the following scheme. Let us equip the set $X$ with another measure $\mu$ and consider the Hilbert space $L^2(X,\ud\mu(x))$ of complex-valued functions on $X$ which are Lebesgue square integrable with respect to $\mu$. One then chooses   an orthonormal set 
$\mathcal{O}$ of functions $\phi_i(x)$ (set aside the question of evaluation map in their respective equivalence classes), satisfying the finiteness and positiveness conditions  
\begin{equation}
	\label{CSNx}
	0< \mathcal{N}(x)= \sum_i \vert\phi_i(x)\vert^2 < \infty\quad (\mathrm{a.e.})\,, 
\end{equation}
and a ``companion''  Hilbert space $\mathcal{H}$ (the space of quantum states)  with orthonormal basis $\{  | e_i\rg \}$ in one-to-one correspondence $\{  | e_i\rg \leftrightarrow \phi_i\}$ with the elements of  $\mathcal{O}$. There results a family $\mathcal{C}$ of unit vectors $|x\rangle$ in $\mathcal{H}$, named, in a generalised sense, coherent states (CS). These vectors  are labelled by elements of $X$ and  resolve the identity operator in $\mathcal{H}$: 
\begin{equation}
	\label{defcsL2}
	X\ni x \mapsto | x\rg =\frac{1}{\sqrt{\mathcal{N}(x)}}\sum_{i}\overline{\phi_i(x)} | e_i\rg\in \mathcal{H}\, . 
\end{equation}
\begin{equation}
	\label{resun}
	\lg x| x\rg = 1\, , \quad \int_{X} \, | x\rg\lg x| \, \ud\nu(x)= \bu_{\mathcal{H}}\, ,\quad  \ud\nu(x)= \mathcal{N}(x)\, \ud\mu(x)\,.
\end{equation}

We now explain the meaning of  the CS quantization within this framework. Let us denote by $\mathbb{P}_{\mathfrak{o}}$ the orthogonal projector mapping  $L^2(X,\ud\mu(x))$ onto the sub-Hilbert space $\mathcal{H}_{\mathfrak{o}}$, which is the closure of the linear span of the elements of $\mathcal{O}$. Clearly, $\mathbb{P}_{\mathfrak{o}}$ is the identity operator on $\mathcal{H}_{\mathfrak{o}}$. Now, we identify $\mathcal{H}_{\mathfrak{o}}\equiv \mathcal{H}$ through a unitary map $U_\mathfrak{o}$ such that we have
\begin{equation}
	\label{identphie}
	U_{\mathfrak{o}}|e_i\rg =\phi_i\,.
\end{equation}
Let $f(x)$ be a function defining the multiplication operator in $L^2(X,\ud\mu(x))$:
\begin{equation}
	\label{defMf}
	M_f: \phi \mapsto f\,\phi \,.
\end{equation}
\begin{prop} (Toeplitz quantization)\label{Toeplitz quantization}
	We have the identity
	\begin{equation}
		\label{MfidAf}
		\mathbb{P}_{\mathfrak{o}}M_f\mathbb{P}_{\mathfrak{o}}\phi_i= U_{\mathfrak{o}}A_fU^{-1}_{\mathfrak{o}}\phi_i \;, \quad\forall i\in\mathbb{N} \;,
	\end{equation}
	where $A_f$ is the CS quantisation of $f$ as
	\begin{equation}
		\label{CSqAf}
		A_f= \int_{X} \,f(x)\, | x\rg\lg x| \, \ud\nu(x)\,. 
	\end{equation}
\end{prop}
\begin{proof}
	We will compare the respective matrix elements of the left- and right-hand sides of \eqref{MfidAf} with respect to the basis $\{|e_i \rg\}$ of $\mathcal{H}$. Let us start with the action of $U_{\mathfrak{o}} A_f U^{-1}_{\mathfrak{o}}$ on $\phi_i$:
	\begin{align*}
		U_{\mathfrak{o}} A_f U^{-1}_{\mathfrak{o}} \phi_i  &=U_{\mathfrak{o}} A_f|e_i\rg = U_{\mathfrak{o}}\int_{X} \,f(x)\, | x\rg\lg x|e_i\rg \, \ud\nu(x) \\
		&=U_{\mathfrak{o}}\sum_{jk}  \int_{X} \,f(x)\, \overline{\phi_j}(x)\, \phi_k(x) \, \ud\mu(x)\, |e_j \rg\lg e_k |e_i\rg\\
		&=U_{\mathfrak{o}} \sum_{j}  \int_{X} \,f(x)\, \overline{\phi_j}(x)\, \phi_i(x) \, \ud\mu(x)\, |e_j\rg\\
		&= \sum_{j}\left(A_f\right)_{ji}\,U_{\mathfrak{o}}|e_j\rg= \sum_{j}\left(A_f\right)_{ji}\,\phi_j.
	\end{align*}
	On the other hand, 
	\begin{align*}
		\mathbb{P}_{\mathfrak{o}}M_f\mathbb{P}_{\mathfrak{o}}\phi_i=\mathbb{P}_{\mathfrak{o}} M_f\phi_i  &= \sum_j \left(\phi_j, f\phi_i\right)_{L^2}\,\phi_j\\
		&=\sum_{j}  \phi_j\,\int_{X}  \overline{\phi_j}(x)\, f(x)\,\phi_i(x) \, \ud\mu(x)\\
		& =   \sum_{j}\left(A_f\right)_{ji}\,\phi_j\,.
	\end{align*}
\end{proof}

Note that $M_\Delta:=M_{\chi_\Delta}$ defines a projection-valued measure such that $M_f=\int f(x)\,\ud M_x$. Indeed, let $g\in L^2(X,\ud\mu(x))$ and $\mu_g(\cdot):=\int_{(\cdot)}|g(x)|^2\,\ud\mu(x)=\left(g|M_{\chi_{(\cdot)}} g\right)_{L^2}$ be the corresponding measure of which $|g|^2$ is the Radon-Nikodym derivative ($\ud\mu_g(x)=|g(x)|^2\,\ud\mu(x)$). Then, for every $g\in L^2(X,\ud\mu(x))$, 
\begin{align*}
	\int f(y)\,\ud\left( g| M_{\chi_{(\cdot)}}g\right)_{L^2}&=\int f(y)\,\ud\mu_g(y)=\int f(y)|g(y)|^2\,\ud\mu(y)\\ &=\left(g| f g\right)_{L^2}=\left( g|M_{f} g\right)_{L^2}  \,. 
\end{align*}
where $\ud\left( g|M_{\chi_{(\cdot)}} g\right)_{L^2}$ denotes integration with respect to the measure $\left( g|M_{(\cdot)}g\right)_{L^2}$.
Moreover, $M_{(\cdot)}:=$ $M_{\chi_{(\cdot)}}$ is a Naimark dilation of the POVM $U_{\mathfrak{o}}A_\Delta U^{-1}_{\mathfrak{o}}:=$ $U_{\mathfrak{o}}A_{\chi_\Delta}U^{-1}_{\mathfrak{o}}=\int_\Delta U_{\mathfrak{o}}|x\rg\lg x| U^{-1}_{\mathfrak{o}}\,\ud\nu(x)$ since, by proposition \ref{Toeplitz quantization}, $\mathbb{P}_{\mathfrak{o}}M_{\chi_\Delta} \mathbb{P}_{\mathfrak{o}}\phi_i=U_{\mathfrak{o}}A_{\chi_\Delta} U^{-1}_{\mathfrak{o}} \phi_i$.

In other words, and with slight abuse of language, every operator $A_f=\int f(x)\,\ud A_x$ is the projection of the operator $M_f=\int f(x)\,\ud M_x$ where $\ud A_x$ denotes integration with respect to the POVM $A_{\chi_\Delta}$ and $\ud M_x$ denotes integration with respect to the PV measure $M_{\chi_\Delta}$. In its turn, $M_{\chi_\Delta}$ is the Naimark's dilation of $A_{\chi_\Delta}$.

An alternative construction of families of $X$-labelled coherent states rests upon unitary transports of one unit-norm state $|\psi\rg$ picked in the Hilbert space $\mathcal{H}$. Let  $\{|e_i\rg\}$ be an orthonormal basis of $\mathcal{H}$ and let $U(x)$ be a family of $X$-labelled unitary operators on $\mathcal{H}$ which obey  square-integrability and orthonormality on $(X,\nu)$  in the sense given in the following.
\begin{prop}
	\label{proptoep}
	Let $\psi $ be arbitrarily picked in $\mathcal{H}$. The family of $X$-labelled unit-norm states in $\mathcal{H}$ defined by 
	\begin{equation}
		\label{defUxpsi}
		X \ni x\mapsto 
		|x\rg_{\psi}:= U(x)|\psi\rg \,,
	\end{equation}
	solves the identity in $\mathcal{H}$ with respect to the measure $\nu$ on $X$,
	\begin{equation}
		\label{UxpsiId}
		\int_{X}  |x\rg_{\psi}{}_\psi\lg x | \,\ud\nu(x)= \bu\,,
	\end{equation}
	if and only if the $X$-labelled family of unitary operators  $U(x)$ obeys square-integrability and orthonormality on $(X,\nu)$  in the following sense
	\begin{align}
		\label{sqintUjx}
		\delta_{ik}&= \int_{X} U_{ji}(x)\,\overline{U_{jk}}(x)\,\ud\nu(x)  \quad \forall\, j\,,\\
		\label{sqintUjlx} 0&=\int_{X}\left[U_{ji}(x)\,\overline{U_{lk}}(x) + U_{li}(x)\,\overline{U_{jk}}(x)\right]  \,\ud\nu(x)\quad \forall\,j,l\,, \ j\neq l\,,
	\end{align}
	where $U_{ij}(x):= \lg e_i|U(x)|e_j\rg$.
\end{prop}
\begin{proof}
	By expanding $\psi= \sum_i a_i\, |e_i\rg$, we have,
	\begin{equation*}
		|x\rg_{\psi}= U(x)|\psi\rg= \sum_{i} a_i\, U(x)\,|e_i\rg = \sum_{ij} a_i\, U_{ji}(x)\,|e_j\rg\,,
	\end{equation*}
	and so 
	\begin{align*}
		|x\rg_{\psi}{}_\psi\lg x | &=  \sum_{ijkl} a_i\bar{a}_k\, U_{ji}(x)\,\overline{U_{lk}}(x)\,|e_j\rg\lg e_l|\\
		&=\sum_{ i k} a_i\bar{a}_k\, \sum_jU_{ji}(x)\,\overline{U_{jk}}(x)\,|e_j\rg\lg e_j| + \sum_{ik} a_i\bar{a}_k\, \sum_{j\neq l}U_{ji}(x)\,\overline{U_{lk}}(x)\,|e_j\rg\lg e_l|\,.
	\end{align*}
	We now integrate the previous expression and impose \eqref{UxpsiId}. From the completeness of the basis, $\sum_j |e_j\rg\lg e_j|= \bu$, and linear independence of the operators $|e_j\rg\lg e_l|$, we infer  
	\begin{align}
		\label{1ikjl}
		1&=\sum_{ik} a_i\bar{a}_k\,  \int_{X}U_{ji}(x)\,\overline{U_{jk}}(x) \,\ud\nu(x) \quad \forall\, j\,,\\
		\label{0ikjl} 0&= \sum_{ik} a_i\bar{a}_k\,  \int_{X}U_{ji}(x)\,\overline{U_{lk}}(x)  \,\ud\nu(x)  \quad \forall \, j, l, \ j\neq l \,,
	\end{align}
	for any unit-norm complex sequence  $(a_i)\in \ell^2$. Let us consider the two operators $\mathrm{M}^{j}$ and  $\mathrm{M}^{j,l}$, $j\neq l$,  with matrix  elements with respect to the basis $\{e_i\}$,
	\begin{equation}
		\label{hermat}
		\mathrm{M}^{j}_{ki}=\int_{X}U_{ji}(x)\,\overline{U_{jk}}(x)  \,\ud\nu(x) \,, \quad \mathrm{M}^{j,l}_{ki}=\int_{X}U_{ji}(x)\,\overline{U_{lk}}(x)  \,\ud\nu(x) \,.
	\end{equation} 
	The  operator $\mathrm{M}^{j}$ is Hermitian and the condition \eqref{1ikjl} means that its mean value obeys $\lg \psi| \mathrm{M}^{j} |\psi\rg =1$ for any normalised $|\psi\rg\in \mathcal{H}$. There results that $\mathrm{M}^{j}=\bu$, and so 
	\begin{equation*}
		\int_{X}U_{ji}(x)\,\overline{U_{jk}}(x)  \,\ud\nu(x)= \delta_{ik}\quad \forall\,j\,.
	\end{equation*}
	The operator $\mathrm{M}^{j,l}$  obeys $\left(\mathrm{M}^{j,l}\right)^{\dag}=\overline{\mathrm{M}}^{l,j}$, and so the operator $\mathrm{M}^{j,l} + \mathrm{M}^{l,j}$ is Hermitian.  The condition \eqref{0ikjl} means that its mean value obeys $\lg \psi| \mathrm{M}^{j,l} + \mathrm{M}^{l,j} |\psi\rg =0$ for any normalised $|\psi\rg\in \mathcal{H}$. There results that $\mathrm{M}^{j,l} + \mathrm{M}^{l,j}=\mathbb{0}$, and so 
	\begin{equation*}
		\int_{X}\left[U_{ji}(x)\,\overline{U_{lk}}(x) + U_{li}(x)\,\overline{U_{jk}}(x)\right]  \,\ud\nu(x)=0\quad \forall\,j,l\,, \ j\neq l\,.
	\end{equation*}
\end{proof}
The above result extends to density operators.
\begin{prop}
	Let $\rho$ be a density operator on $\mathcal{H}$, and $\rho(x):= U(x)\,\rho\,U^{\dag}(x)$  its unitary transport   $U(x)$ which has square-integrability and orthonormality properties \eqref{sqintUjx} and  \eqref{sqintUjlx}. We then have the resolution of the identity:
	\begin{equation}
		\label{UxrhoUxiId}
		\int_{X} \rho(x)\,\ud\nu(x)= \bu\,.
	\end{equation}
\end{prop}
\begin{proof}
	From the decomposition of   $\rho$ in a sum of rank-one projectors, 
	\begin{equation*}
		\rho= \sum_i p_i\, |\psi_i\rg\lg\psi_i|\,, \quad \Vert \psi_i\Vert^2=1\,\quad \sum_i p_i =1\,, \quad 0\leq p_i \leq 1\,,
	\end{equation*}
	we have 
	\begin{equation*}
		\rho(x)= \sum_i p_i\, |x\rg_{\psi_i}{}_{\psi_i}\lg x |\,.
	\end{equation*}
	By applying \eqref{UxpsiId} to each CS projectors in this sum, we obtain 
	\begin{equation*}
		\int_{X} \rho(x)\,\ud\nu(x)= \sum_i p_i \,\bu = \bu\,.
	\end{equation*}
\end{proof}
Let us now understand the Toeplitz quantization within this alternative approach. Implementing  the construction  \eqref{defUxpsi} with the choice $\psi = e_j$ \underline{and} $U^{\dag}(x)$, we obtain
\begin{equation}
	\label{Uej}
	|x\rg_{e_j}= \sum_i \overline{U_{ji}}(x)\,|e_i\rg\,.
\end{equation}
We denote by $\mathcal{O}_j=\{U_{ji} \}$ the orthonormal system in $L^2(X,\ud \nu(x))$ formed by the functions $U_{ji}(x)$ at fixed $j$ according to the condition \eqref{sqintUjx}, and by $\mathbb{P}_{\mathfrak{o}_j}$ the corresponding orthogonal projector mapping  $L^2(X,\ud\nu(x))$ onto the sub-Hilbert space $\mathcal{H}_{\mathfrak{o}_j}$ which is the closure of the linear span of the elements of $\mathcal{O}_j$. With the identification \eqref{identphie}, \textit{i.e.} 
$U_{\mathfrak{o}_j}|e_i\rg = U_{ji}$, relation \eqref{MfidAf} holds under the form,
\begin{equation}
	\label{PMAPj}
	\mathbb{P}_{\mathfrak{o}_j}M_f\mathbb{P}_{\mathfrak{o}_j}U_{ji}= U_{\mathfrak{o}_j}A^{(j)}_{f} U^{-1}_{\mathfrak{o}_j}U_{ji}\,, \quad A^{(j)}_{f}= \int_{X}f(x) \,|x\rg_{e_j}{}_{e_j}\lg x|\,\ud\nu(x)\,.
\end{equation}
We now extend this result to the integral quantization with a density operator $\rho$ on $\mathcal{H}$. We choose the orthonormal basis $\{|e_i\rg\}$ of the latter as made of the eigenstates of $\rho$ (with possible null eigenvalue), 
\begin{equation}
	\label{rhoei}
	\rho= \sum_i p_i\, |e_i\rg\lg e_i|\,, \quad \sum_i p_i =1\,, \quad 0\leq p_i \leq 1\,.
\end{equation}
Clearly we have, 
\begin{equation}
	\label{AfAfj}
	A^{(\rho)}_{f}= \int_{X}f(x) \,\rho(x)\,\ud\nu(x)= \sum_j p_j\int_{X}f(x) \,|x\rg_{e_j}{}_{e_j}\lg x|\,\ud\nu(x)= \sum_j p_j A^{(j)}_{f}\,.
\end{equation}
By linearity, from \eqref{PMAPj} we infer
\begin{equation}
	\label{PMAPrho}
	\sum_j p_j\mathbb{P}_{\mathfrak{o}_j}M_f\mathbb{P}_{\mathfrak{o}_j}= \sum_j p_j  U_{\mathfrak{o}_j}A^{(j)}_{f}U^{-1}_{\mathfrak{o}_j} \,. %= A^{(\rho)}_{f}\sum_j \mathbb{P}_{\mathfrak{o}_j}\,,
\end{equation}
Let us complete \eqref{sqintUjx} by imposing the full orthonormality constraint on the matrix elements $U_{ij}(x)$:
\begin{align}
	\label{sqintUjix}
	\delta_{jl}\delta_{ik}&= \int_{X} U_{ji}(x)\,\overline{U_{lk}}(x)\,\ud\nu(x)  \quad \forall\,i, j,k,l\,.
\end{align}
It follows that  $\mathbb{P}_{\mathfrak{o}_j}\mathbb{P}_{\mathfrak{o}_l}= 0$ $\forall j\neq l$.  Moreover, \eqref{AfAfj} becomes the direct sum
\begin{equation}
	\label{AfAfjdir}
	A^{(\rho)}_{f}= \bigoplus_j p_j A^{(j)}_{f}\,. 
\end{equation}
Defining  the  direct sum of unitary operators $U_{\mathfrak{o}}=  \bigoplus_j U_{\mathfrak{o}_j}$ we write \eqref{PMAPrho} as
\begin{equation}
	\label{}
	\bigoplus_j p_j\mathbb{P}_{\mathfrak{o}_j}M_f \mathbb{P}_{\mathfrak{o}_j}= \bigoplus_j p_j  U_{\mathfrak{o}_j}A^{(j)}_{f}U_{\mathfrak{o}_j}^{-1}
	= U_{\mathfrak{o}} A^{(\rho)}_{f} U^{-1}_{\mathfrak{o}}
\end{equation}
where we recover an analogous result about Naimark dilation in the case of density matrices.

%%%%%%%%%%%%%%%%%%%%%%%%%%%%%%%%%%%%%%%%%%%%%%%%%%%%%%%%%%%%%%%%

\section{Integral POVM quantization of fonctions on the circle}
\label{intquant}

Let us now describe the integral quantization for building observables acting on the two-dimensional real Hilbert space, \textit{i.e.}, the Euclidean plane. In this section, we give an overview of the material developed in  \cite{bergeron2019orientations} on the integral quantization of functions on  the unit circle. We describe  rays  in the Euclidean plane, and real POVMs built from density matrices in dimension 2. We then discuss the integral quantization procedure based on the group of rotations SO$(2)$.

\subsection{Euclidean plane with Dirac notations}
\label{dirnotR2}

First, a quantum state can be decomposed into a sum of pure states, each identified with a unit vector. In the specific case of $\R^2$, the polar angle $\phi \in [0,2\pi)$ is associated with the unit vector $\widehat{\mathbf{u}}_{\phi}$. This unit vector corresponds to the pure state $|\phi\rg := \left| \widehat{\mathbf{u}}_{\phi}\right\rg$. In particular, the two axes of the Euclidean plane can be described as a ``horizontal orientation" state $\widehat{\boldsymbol{\imath}} = |0\rg$ and a ``vertical orientation" state $\widehat{\boldsymbol{\jmath}}=\left| \dfrac{\pi}{2}  \right\rangle$. The plane is depicted in figure \ref{figure quantum states}. 

\begin{figure}
\begin{center}
\setlength{\unitlength}{0.15cm} 
\begin{picture}(60,60)
\put(10,20){\vector(1,0){30}} 
\put(10,20){\vector(0,1){30}} 
%\put(30,20){\circle{}} 
\put(43, 24){\makebox(0,0){$\widehat{\boldsymbol{\imath}} = |0\rangle \equiv \begin{pmatrix}
      1    \\
      0  
\end{pmatrix}$}}
%\put(43.6,47.2){\circle*{2}}
\put(29, 52){\makebox(0,0){$|\phi\rangle = \begin{pmatrix}
      \cos\phi    \\
      \sin\phi  
\end{pmatrix}$}}  
\put(45, 41){\makebox(0,0){\footnotesize (quantum) state in ``orientation representation''}} 
\put(10, 16.5){\makebox(0,0){$O$}} 
\put(30, 10){\makebox(0,0){$\langle  0 | 0  \rangle = 1 =\left\langle \frac{\pi}{2} \right| \left.\frac{\pi}{2}\right\rangle\, , \quad \langle 0 \left| \frac{\pi}{2} \right\rangle = 0$}}
\put(30, 1){\makebox(0,0){$\bu = | 0  \rangle \langle  0 |  +  \left| \frac{\pi}{2} \right\rangle \left\langle  \frac{\pi}{2} \right| \ \Leftrightarrow \ \left( \begin{array}{cc}
1 & 0\\ 0 & 1 \end{array} \right) = \left( \begin{array}{cc}
1 & 0\\ 0 & 0 \end{array} \right) + \left( \begin{array}{cc}
0 & 0\\ 0 & 1 \end{array} \right)$}} 
\put(8, 52){\makebox(0,0){$\widehat{\boldsymbol{\jmath}} = |\pi/2\rangle =\begin{pmatrix}
      0    \\
      1  
\end{pmatrix}$}} 
%\put(46, 42){\makebox(0,0){$\vec u_r$}} 
\put(16, 38){\makebox(0,0){$1$}} 
%\put(12, 36){\makebox(0,0){$\vec u_{\theta}$}}
%\qbezier(34,20)(32,25)(3,3)
%\put(32.3, 22){\makebox(0,0){$\circlearrowleft$}}
\put(20, 26){\makebox(0,0){$\phi$}}
\put(14,20){\oval(10,15)[tr]}
\thicklines 
\put(10,20){\vector(1,2){14}} 
%\put(30,20){\vector(-2,1){28}}
\label{figure quantum states}
\end{picture}
\end{center}
\caption{The Euclidean plane and its unit vectors viewed as pure quantum states in Dirac ket notations.}
\label{R2fig}
\end{figure}

Therefore, a pure state in the horizontal-vertical representation can be decomposed as
\begin{equation}
\label{Fock}
| \phi \rangle = \cos{\phi} \,| 0 \rangle  + \sin{\phi}\, \left| \frac{\pi}{2} \right\rangle\, , \quad \lg 0 | \phi\rg = \cos \phi\, , \qquad \left\lg \frac{\pi}{2} \right|\phi\rg =  \sin\phi\, ,
\end{equation}
with the associated probability distributions $\phi\mapsto \cos^2 \phi$  and $\phi\mapsto \sin^2 \phi$ on the unit circle equipped with the measure $\ud\phi/\pi$. 
 One also have  the  unsharp \textit{unit circle} representation of a pure state given by the trigonometric function
\begin{equation}
\label{uncirrep}
\lg \eta | \phi\rg = \cos (\phi -\eta)\, .
\end{equation}
with associated probability distribution families 
\begin{equation}
\label{probdistuncirc}
\eta\mapsto \mathcal{P}_{\phi}(\eta) = \vert \lg \eta| \phi\rg\vert^2= \cos^2(\phi-\eta)\, , \quad \int_0^{2\pi} \mathcal{P}_{\phi}(\eta) \, \frac{\ud \eta}{\pi} = 1\, . 
\end{equation}
To  $| \phi\rangle$ corresponds the orthogonal projector  $E_{\phi} = | \phi \rangle \langle \phi |$  (also called pure state),
\begin{align}
\label{projtheta} 
E_{\phi} = \begin{pmatrix}
 \cos{\phi} \\
 \sin{\phi}\end{pmatrix} \begin{pmatrix}
\cos{\phi} &\sin{\phi} \end{pmatrix}  = \begin{pmatrix}
\cos^2{\phi} & \cos{\phi} \sin{\phi}  \\
\cos{\phi} \sin{\phi}  & \sin^2{\phi} \end{pmatrix} 
=\mathcal{R}(\phi) |0\rg\lg 0|\mathcal{R}(-\phi) \, , 
\end{align}
where $\mathcal{R}(\phi)=\begin{pmatrix}
  \cos\phi    &  -  \sin\phi  \\
   \sin\phi   &   \cos\phi 
\end{pmatrix}$.
In contrast to pure states the mixed states are  defined by two-dimensional density matrices of the form
\begin{equation}
\label{dens-a-bM}
\rho := \mathsf{M}(a,b) = \begin{pmatrix}
   a   &  b \\
    b  &  1-a
\end{pmatrix}\,,  \quad 0\leq a\leq 1\, , \quad \Delta:= \det \rho = a(1-a)-b^2 \geq 0\, . 
\end{equation}

From the spectral decomposition of $\rho$
\begin{equation}
\label{specrho}
\rho = \lambda E_{\phi} + (1-\lambda) E_{\phi + \pi/2}\, , 
\end{equation}
we only have to consider the upper half-plane, \textit{i.e.},  $0\leq \phi \leq \pi$. The highest eigenvalue of $\rho$ is $\lambda= \frac{1}{2}(1+\sqrt{1-4\Delta})$, such that $1/2\leq \lambda \leq 1$.  
Defining $r:= 2\lambda -1$, $0\leq r\leq 1$,  
a real density matrix takes the form
\begin{align}
	\label{specrhor}
	\rho = \left(\frac{1+r}{2}\right) E_{\phi} + \left(\frac{1-r}{2}\right) E_{\phi+\pi/2} \;,
\end{align}
and, in polar coordinates $(r,\phi)$, we obtain
\begin{equation}
\label{standrhomain}
\rho\equiv\rho_{r,\phi}=\frac{1}{2} \bu + \frac{r}{2}\mathcal{R}(\phi)\begin{pmatrix}
  1    &  0  \\
  0    &  -1
\end{pmatrix}\mathcal{R}(-\phi)= \begin{pmatrix}
  \frac{1}{2}  + \frac{r}{2}\cos2\phi  &   \frac{r}{2}\sin2\phi  \\
\frac{r}{2}\sin2\phi    &   \frac{1}{2}  - \frac{r}{2}\cos2\phi
\end{pmatrix}\, . 
\end{equation}
In particular, we retrieve the projector on the unit vector, $ \rho_{1,\phi}= E_{\phi}$.

Note that the parameter $r$ encodes  the distance of $\rho$ to the pure state $E_{\phi}$ while $1-r$ measures the degree of ``mixing" \cite{bergeron2019orientations}. Its statistical interpretation is given by the von Neumann entropy  defined as
\begin{equation}
\label{VNentr}
S_{\rho}:= -\mathrm{Tr}(\rho\,\ln\rho)= -\lambda \ln\lambda - (1-\lambda)\ln(1-\lambda)= -\frac{1+r}{2} \ln\frac{1+r}{2} - \frac{1-r}{2} \ln\frac{1-r}{2}\, .  
\end{equation}
As a function of $\lambda \in [1/2,1]$, $S_{\rho}$ is non-negative, concave and symmetric with respect to its maximum value $\log 2$ at $\lambda=1/2$, which corresponds to $r=0$, \textit{i.e.}, $\rho_0\equiv \bu/2$, which describes the state of  completely random orientations.

\subsection{Integral POVM quantization}
\label{sec:integquant}
The measure space $(X,\ud \nu(x))$ for the Euclidean plane is the unit circle with its uniform (Lebesgue) measure:
\begin{equation}
\label{measS1}
X= \mathbb{S}^1\, , \quad \mathrm{d}\nu(x) = \frac{\ud \phi}{\pi}\, , \quad \phi \in [0, 2\pi)\, .
\end{equation}
$\mathbb{S}^1$ is here viewed  as the abelian group SO$(2)$ of rotations in the plane. As a consequence, the density operator \eqref{standrhomain} is manifestly covariant under rotation
\begin{equation}
\label{rotrhomain}
 \rho_{r,\phi_0}(\phi)= \mathcal{R}\left(\phi\right) \rho_{r,\phi_0}\mathcal{R}\left(-\phi\right)=  \rho_{r,\phi_0+\phi}\, , \quad 0\leq \phi< 2\pi\, . 
\end{equation}
Additionally, one easily proves by elementary integration on matrix elements that  this family solves the identity 
\begin{equation}
\label{margomegamain}
\int_0^{2\pi} \rho_{r,\phi+\phi_0} \, \frac{\ud\phi}{\pi}= \bu\,.
\end{equation}
This property allows to apply the procedure \eqref{iqmap} to quantize a function (or distribution) $f(\phi)$ on the circle. Its quantum counterpart  is  the 2$\times$2 matrix operator
\begin{equation}
\label{qtfrhor}
\begin{split}
f \mapsto A_f &= \int_0^{2\pi} f(\phi) \rho_{r,\phi+\phi_0} \, \frac{\ud\phi}{\pi}
= \begin{pmatrix}
  \lg f\rg  + \frac{r}{2}C_c\left(R_{\phi_0}f\right)  &   \frac{r}{2}C_s\left(R_{\phi_0}f\right) \\
\frac{r}{2}C_s\left(R_{\phi_0}f\right)   &   \lg f\rg - \frac{r}{2}C_c\left(R_{\phi_0}f\right)
\end{pmatrix}\\
&= \lg f\rg \,I + \frac{r}{2}\left[C_c\left(R_{\phi_0}f\right)  \,\sigma_3 + C_s\left(R_{\phi_0}f\right) \, \sigma_1\right]\,,
\end{split}
\end{equation}
where $ \lg f\rg:= \frac{1}{2\pi}\int_0^{2\pi}f(\phi)\,\ud\phi$ is the average of $f$ on the unit circle and $R_{\phi_0}(f)(\phi) := f(\phi-\phi_0)$. The symbols $C_c$ and $C_s$ are for the cosine and sine doubled angle Fourier coefficients of $f$,
\begin{equation}
\label{CcCs}
C_c(f) = \int_0^{2\pi} f(\phi) \cos2\phi \, \frac{\ud\phi}{\pi}\, , \quad C_s(f) = \int_0^{2\pi} f(\phi) \sin2\phi \, \frac{\ud\phi}{\pi}\, . 
\end{equation}

We observe that the quantisation of the components $(\cos(\phi-\phi_0), \sin(\phi-\phi_0))$ of the density matrix eigenvectors give $0$. Let us denote by $N_{\phi_0}$ the space of functions or distributions (in a certain sense) $T(\phi)$ on the circle which obey 
\begin{equation}
\label{Nphio}
 \int_0^{2\pi} T(\phi) \rho_{r,\phi+\phi_0} \, \frac{\ud\phi}{\pi} = 0\,.
\end{equation}
Such a space includes the orthogonal in $L_{\R}^2(\mathbb{S}^1,\ud\phi/\pi)$ to the $3$-dimensional subspace spanned by the functions $1$, $\cos2\phi$, and $\sin2\phi$. Thus,  the operator $A_f$ is the $\rho_{r,\phi+\phi_0}$-quantization of $f$ modulo any element of  $N_{\phi_0}$. Actually, by identifying the Euclidean space $\R^3$ with the real span $V_3$ of the orthonormal set in $L_{\R}^2(\mathbb{S}^1,\ud\phi/\pi)$,
\begin{equation}
\label{3basis}
V_3 = \mathrm{Span}\left\{e_0(\phi):=\frac{1}{\sqrt{2}}, e_1(\phi):=\cos2\phi, e_2(\phi):=\sin2\phi \right\}\,,
\end{equation}
the linear map \eqref{qtfrhor} yields a non-commutative version of $\R^3$. In particular, we have
\begin{align}
\label{q3basis0}
A_{e_0} &= \frac{\bu}{\sqrt{2}}\,,\\
 \label{q3basis1}A_{e_1}&= \frac{r}{2}[\cos 2\phi_0\,\sigma_3 + \sin 2\phi_0\,\sigma_1]\equiv  \frac{r}{2} \sigma_{2\phi_0}\,, \\
  \label{q3basis2} A_{e_2}&= \frac{r}{2}[-\sin 2\phi_0\,\sigma_3 + \cos 2\phi_0\,\sigma_1]\equiv  \frac{r}{2} \sigma_{2\phi_0+ \pi/2}\,,
\end{align}
with the non-zero commutation rule
\begin{equation}
\label{com12}
\left[A_{e_1},A_{e_2}\right]= -\frac{r^2}{2}\tau_2\, , \quad \tau_2:= \begin{pmatrix}
   0   &  -1  \\
    1  &  0
\end{pmatrix}\,. 
\end{equation}
In the above we have associated  with the unit vector $\widehat{\mathbf{u}}_{\phi}$ \eqref{Fock} the  symmetric matrix,
\begin{equation}
\label{sigphi}
\sigma_{\phi}:=  \cos \phi\, \sigma_{3} +  \sin \phi\, \sigma_1\equiv \overrightarrow{\boldsymbol{\sigma}}\cdot \widehat{\mathbf{u}}_{\phi}= \begin{pmatrix}
  \cos \phi    &  \sin \phi  \\
 \sin \phi     & - \cos \phi
\end{pmatrix} = \mathcal{R}(\phi)\,\sigma_3\, ,
\end{equation}
where we have introduced the  vector notation $\overrightarrow{\boldsymbol{\sigma}}= \sigma_3 \widehat{\boldsymbol{\imath}}+ \sigma_1\widehat{\boldsymbol{\jmath}}$. Note the commutation rule
\begin{equation}
\label{comppp}
\left[\sigma_{\phi},\sigma_{\phi^{\prime}}\right]= 2\sin(\phi-\phi^{\prime})\tau_2\,. 
\end{equation}
Also note that the $3$ matrices $\sigma_1 $, $\tau_2 $, $\sigma_3 $ span the Lie algebra $\mathfrak{sl}$(2,$\R$), while $\tau_2 $ is precisely the generator of rotations in the plane:
\begin{equation}
\label{genrottau}
\mathcal{R}(\phi) = e^{\tau_2\phi}\,.
\end{equation}
The spectral decomposition of the quantum observable $\sigma_{2\phi}$ precisely involves its eigen-orientations $\left|\phi\right\rg$ (eigenvalue $1$) and $\left| \phi+ \frac{\pi}{2}\right\rg$ (eigenvalue $-1$), 
\begin{equation}
\label{sig2phi}
\sigma_{2\phi} = \left|\phi\right\rg\left\lg\phi\right| - \left| \phi+ \frac{\pi}{2}\right\rg\left\lg \phi + \frac{\pi}{2}\right|=
 E_{\phi} - E_{\phi + \pi/2} \, .
\end{equation}
 In agreement with \eqref{specrho}, the density operator \eqref{standrhomain} is expressed in terms of $\sigma_{2\phi}$ as
\begin{equation}
\label{standrhomain1}
\rho_{r,\phi}= \frac{1}{2}\left( \bu + r \sigma_{2\phi}\right)\, . 
\end{equation}
Finally, note that the covariance requirement emerges from the rotational invariance of the density operator \eqref{rotrhomain}
\begin{equation}
\label{covpropcircle}
 \mathcal{R}\left(\theta\right) \,A_f\,\mathcal{R}\left(-\theta\right)= A_{R_{\theta}(f)}\,. 
\end{equation}
An interesting question concerns the existence of functions $f(\phi)$ in $V_3$  yielding density matrices through the quantization map \eqref{qtfrhor}. Solving this problem amounts to find $f(\phi) = a_0 + a_2 \cos2\phi + b_2 \sin2\phi$ such that $A_f = \rho_{s,\theta}$. Solutions are given from the relations
\begin{equation}
\label{frho}
a_0 = \frac{1}{2}\, , \quad
\left\lbrace\begin{array}{c}
   r\left(a_2\cos2\phi_0 -b_2 \sin2\phi_0\right)  =   s\cos2\theta  \\
   r\left(a_2\sin2\phi_0 + b_2 \cos2\phi_0\right)  =   s\sin2\theta
\end{array}\right.\,, \quad 0\leq r\,,\,s\leq 1\, .
\end{equation}
Putting $a_2 = l \cos2\theta_0$ and $b_2 = l \sin2\theta_0 $ with $l>0$, we obtain from above,
\begin{equation}
\label{frho1}
l = \frac{s}{r}\, , \quad \theta_0= \theta -\phi_0\,,
\end{equation}
which yields the solution
\begin{equation}
\label{frho2}
f(\phi)= \frac{1}{2} +\frac{s}{r}\,\cos2(\phi + \phi_0-\theta)\,. 
\end{equation}
The representation of a given mixed state as a continuous superposition of mixed states as
\begin{equation}
\label{frho3}
\rho_{s,\theta}= \int_0^{2\pi} \left[\frac{1}{2} +\frac{s}{r}\, \cos2\phi\right]\, \rho_{r,\phi + \theta} \, \frac{\ud\phi}{\pi} \,,
\end{equation}
which is convex for $r\geq 2s$, provides one more  illustration of a \textit{typical property of quantum-mechanical ensembles in comparison with their classical counterparts.} Note that if $r\geq 2s$ holds,  $\phi \mapsto \left[\frac{1}{2} +\frac{s}{r}\, \cos2\phi\right]$ is a  probability distribution and \eqref{frho3} can also be viewed as the average of these mixed states.

\subsection{Quantization and symbol calculus}

Here, we go forward in the algebraic and geometrical interpretation of the quantisation map \eqref{qtfrhor}, yielded by the POVM built from the   family of density matrices 
\begin{equation}
\label{famropovm}
\left\{\rho_{r,\phi+\phi_0}\,,\, 0\leq \phi< 2\pi\, \mathrm{mod}\,2\pi\right\}\,,
\end{equation}
with parameters $0\leq r\leq 1$ and $0\leq \phi_0< 2\pi\, \mathrm{mod}\,2\pi$, and restricted to the $3$-d Fourier realisation $V_3$ of the Euclidean space $\R^3$. Given a  vector $\vec{f} = f_1\bi + f_2\bj +f_0 \bk \in \R^3$, identified with $\vec{f} = f_1e_1 + f_2 e_2+f_0 e_0\in V_3$, its quantization is the quantum observable
\begin{equation}
\label{fAfsigma}
A_f =  f_1A_{e_1} + f_2 A_{e_2}+f_0 A_{e_0}=   f_1  \frac{r}{2} \sigma_{2\phi_0}  +  f_2\frac{r}{2} \sigma_{2\phi_0+\pi/2} + f_0 \frac{\bu}{\sqrt{2}}\,.
\end{equation}
Three basic matrices, namely the identity $\bu$ and the two real Pauli matrices $\sigma_1$ and $\sigma_3$,
 generate the \index{Jordan algebra} Jordan algebra $\mathcal{J}_3$ \footnote{A Jordan algebra is an (not necessarily associative) algebra over a field whose multiplication is commutative and satisfies the \emph{Jordan identity}:
$(xy)x^2 = x(yx^2)$ .} of all real symmetric $2\times2$ matrices. 
Any element  $A$ of the algebra  decomposes as
\begin{equation} 
A \equiv  \left( \begin{array}{cc}
a& b \\  b&  d \end{array}  \right) = \frac{a+d}{2}I + \frac{a-d}{2} \sigma_3 + b \sigma_1 
\equiv \alpha I + \delta \sigma_3 + \beta \sigma_1\, . \label{msym13}
\end{equation}
The product in this algebra is defined by
\begin{equation}
\mathcal{O''} = A\odot A' = \frac{1}{2} \left( A A' + A' A\right)\, ,
\label{jord13}
\end{equation}
which entails on the level of components $\alpha, \delta, \beta$, the relations : 
\begin{equation}
\alpha'' = \alpha \alpha' + \delta \delta' + \beta \beta', \  
\delta'' = \alpha \delta' + \alpha' \delta,\  \beta'' = \alpha \beta'+ \alpha' \beta\, .
\end{equation}
In particular we have
\begin{equation}
\sigma_1 \odot \sigma_3 = \sigma_3 \odot \sigma_1= 0\,, \quad \mbox{while}\quad \sigma_1  \sigma_3 = -\sigma_3 \sigma_1= \tau_2=\begin{pmatrix}
   0   &  -1  \\
   1   &  0
\end{pmatrix}\,.
\end{equation}
Equipped with the scalar product
\begin{equation}
\label{inprod}
\lg A|B\rg_{\mathcal{J}_3} = \mathrm{Tr}(AB)
\end{equation}
the algebra $\mathcal{J}_3$ is also a $3$-d real Hilbert space. To any $\phi \in [0,2\pi)\, \mathrm{mod}\, 2\pi $ corresponds the orthonormal basis 
\begin{equation}
\label{orbassigma} 
 \left\lbrace \frac{\sigma_\phi}{\sqrt{2}}\,, \quad  \frac{\sigma_{\phi+\pi/2}}{\sqrt{2}}\,, \quad \frac{\bu}{\sqrt{2}} \right\rbrace \,. 
\end{equation}
It results from \eqref{fAfsigma} that, for $0<r\leq1$, we have 
\begin{equation}
\label{fnorm2Af}
\Vert \vec{f}\Vert^2_{V_3}= f_1^2+f_2^2 + f_0^2 = \lg A_f|\Delta_rA_f\rg_{\mathcal{J}_3}\,, \quad \Delta_r:=\mathrm{diag}\left( \frac{2}{r^2},  \frac{2}{r^2} ,1\right)\, ,
\end{equation}
where $A_f$ is viewed as the vector
\begin{equation}
	A_f = \begin{pmatrix}
	\frac{f_1 r}{\sqrt{2}}       \\
	\frac{f_2 r}{\sqrt{2}}    \\  
	f_0
\end{pmatrix}
\end{equation}
in the $3$-d real Hilbert space with basis \eqref{orbassigma}, as shown in \eqref{fAfsigma}. From the positiveness of $\Delta_r$, we infer that the map $ \R^3 \ni \vec{f} \mapsto A_f\in \mathcal{J}_3$ is one-to-one but not unitary. For this reason, any symmetric matrix $A$ is, for $r\neq 0$, the $\rho_{r,\phi+\phi_0}$-quantization of a unique function $f_A \in \R^3$ as
\begin{equation}
\label{AfA}
\begin{split}
A\equiv A_{f_A} &=a_1  \frac{1}{\sqrt{2}} \sigma_{2\phi_0}  +  a_2 \frac{1}{\sqrt{2}} \sigma_{2\phi_0+\pi/2} + a_0 \frac{\bu}{\sqrt{2}}=\int_0^{2\pi}\frac{\ud \phi}{\pi}\,f_A(\phi)\,\rho_{r,\phi+\phi_0}\,,\\ 
f_A(\phi)&= 
a_1\frac{\sqrt{2}}{r}\cos 2\phi+a_2\frac{\sqrt{2}}{r}\sin 2\phi + \frac{a_0}{\sqrt{2}}\,.
\end{split}
\end{equation}
Adopting the terminology of Berezin-Lieb mentioned in Section \ref{semclasspor}, $f_A$ is  the upper (Lieb) or contravariant (Berezin) symbol $\check{A}$ of $A$. Similarly, the lower (Lieb) or covariant (Berezin) symbol of a symmetric matrix $A=\begin{pmatrix}
 a     &   b \\
  b    &  d
\end{pmatrix}$ stands for the semi-classical portrait of the latter. It is here defined coherently with \eqref{lowsymbgen}:
\begin{equation}
\label{lwcircle}
 \check{A}(\phi)=  \mathrm{Tr}\left(A\rho_{r,\phi_0}(\phi)\right) = \frac{a+d}{2} +r\left(\frac{a-d}{2}\cos2(\phi+\phi_0) + b\sin2(\phi+\phi_0)\right)\,.
\end{equation}  
 If $A\equiv A_f$ is the $\rho_{r,\phi_0}$-quantized of $f\in V_3$, we obtain 
 \begin{equation}
\label{lwAf}
\check{A}_f(\theta)\equiv \check{f}(\theta)= \lg f\rg + rs\int_0^{2\pi}\frac{\ud \phi}{\pi} \,f(\phi)\, \cos2\left(\phi +\phi_0 - \theta -\theta_0\right)\,.
\end{equation}

Applying this to each component of $f$, we show the map $f\mapsto \check{f}$ in $V_3$ corresponds to the matrix transformation in $\R^3$:
\begin{align}
	\label{ffcheck}
	\overrightarrow{\check{f}}=\begin{pmatrix}
		\check{f}_1      \\
		\check{f}_2 \\
		\check{f}_0 
	\end{pmatrix}= \begin{pmatrix}
		rs\cos2\left(\theta_0-\phi_0\right)    & - rs\sin2\left(\theta_0-\phi_0\right)  &0\\
		rs\sin2\left(\theta_0-\phi_0\right)   & rs\cos2\left(\theta_0-\phi_0\right) & 0\\
		0  & 0& 1
	\end{pmatrix}\begin{pmatrix}
		f_1    \\
		f_2  \\
		f_0
	\end{pmatrix}\,,
\end{align}
describing a rotation of angle $2 \left(\theta_0 - \phi_0\right)$ around the direction $f_0$ together with a possible contraction of the components $f_1$ and $f_2$ by a factor $rs$ (remember $0\leq r,s\leq1$).

The upper or contravariant symbol  $\hat{A}(\phi)$ appearing in the operator-valued integral
\begin{equation}
A = \frac{1}{\pi} \int_0^{2\pi} d\phi \, \hat{A}(\phi) | \phi  \rangle
\langle \phi |\, . \label{ssup13}
\end{equation}
is highly non-unique, but is chosen as the simplest one.

\subsection{Integral POVM quantization on the circle and  Toeplitz-Naimark formalism}
Let us see how the  material introduced in Section \ref{toepnaim} fits the quantization map \eqref{qtfrhor}. In the present case the large Hilbert space is the whole $L_{\R}^2(\mathbb{S}^1,\ud\phi/\pi)$, i.e., the Hilbert space of real-valued functions which are square-integrable on the circle w.r.t. the measure $\ud\phi/\pi$. Now consider the quantization of a function $f(\phi)$ on the circle,  as is described in Eq. \eqref{qtfrhor}. We start from the linear operator in $L_{\R}^2(\mathbb{S}^1,\ud\phi/\pi)$ formally defined by
\begin{equation}
	\label{MfV5}
	L_{\R}^2(\mathbb{S}^1,\ud\phi/\pi)\ni v \mapsto M_f v = fv\,.
\end{equation}
and valid if $fv \in L_{\R}^2(\mathbb{S}^1,\ud\phi/\pi)$.
We now consider the family of  rotation matrices  $ \mathcal{R}\left(\phi\right)$ acting on the plane $\mathcal{H}=\R^2$. 
The matrix elements $$ \left\{\mathcal{R}_{ij}\left(\phi\right)\right\}=  \left\{\mathcal{R}_{11}\left(\phi\right)= \mathcal{R}_{22}\left(\phi\right)=\cos\phi, \mathcal{R}_{12}\left(\phi\right)= -\mathcal{R}_{21}\left(\phi\right)=-\sin\phi\right\} $$ satisfy the conditions of square-integrability and orthonormality \eqref{sqintUjx} and \eqref{sqintUjlx} on the circle. 
\begin{align}
	\label{sqintRjx}
	\delta_{ik}&= \int_0^{2\pi} \mathcal{R}_{ji}(\phi)\,\mathcal{R}_{jk}(\phi)\,\frac{\ud \phi}{\pi}  \quad \forall\, j\,,\\
	\label{sqintRjlx} 0&= \int_0^{2\pi}\left[ \mathcal{R}_{ji}(\phi)\, \mathcal{R}_{lk}(\phi) +  \mathcal{R}_{li}(\phi)\, \mathcal{R}_{jk}(\phi)\right]  \,\frac{\ud \phi}{\pi} \quad \forall\,j,l\,, \ j\neq l\,.
\end{align}
Explicitly,
\begin{align}
	\label{sqintR1x}
	1&= \int_0^{2\pi} \left(\mathcal{R}_{11}(\phi)\right)^2\,\frac{\ud \phi}{\pi} =  \int_0^{2\pi} \left(\mathcal{R}_{12}(\phi)\right)^2\,\frac{\ud \phi}{\pi}\\
	\label{sqintR2x}
	0&= \int_0^{2\pi} \mathcal{R}_{11}(\phi)\,\mathcal{R}_{12}(\phi)\,\frac{\ud \phi}{\pi} \\
	\label{sqintR11x} 0&= \int_0^{2\pi}\left[ \mathcal{R}_{11}(\phi)\, \mathcal{R}_{21}(\phi) +  \mathcal{R}_{21}(\phi)\, \mathcal{R}_{11}(\phi)\right]  \,\frac{\ud \phi}{\pi} ,\\
	\label{sqintR12x} 0&= \int_0^{2\pi}\left[ \mathcal{R}_{11}(\phi)\, \mathcal{R}_{22}(\phi) +  \mathcal{R}_{21}(\phi)\, \mathcal{R}_{12}(\phi)\right]  \,\frac{\ud \phi}{\pi} \,,
\end{align}
and so on. Within the framework established with Proposition \ref{proptoep},  let us identify the two orthonormal systems $\mathcal{O}_j\subset L_{\R}^2(\mathbb{S}^1,\ud\phi/\pi)$, $j=1,2$. They are 
\begin{equation}
	\label{OURij}
	\mathcal{O}_1= \{\cos\phi, \sin\phi\} =  \{v_1(\phi), v_2(\phi)\}\,\quad \mathcal{O}_2= \{-\sin\phi, \cos\phi\} =  \{-v_2(\phi), v_1(\phi)\} \,.
\end{equation}

We denote by $\mathbb{P}_{\mathfrak{o}_j}$ the orthogonal projector mapping  $L_{\R}^2(\mathbb{S}^1,\ud\phi/\pi)$ onto the sub-Hilbert space $\mathcal{H}_{\mathfrak{o}_j}$ which is  the linear span of the elements of $\mathcal{O}_j$. Then we identify $\mathcal{H}_{\mathfrak{o}_j}\equiv \mathcal{H}$ through
\begin{equation}
	U_{\mathfrak{o}_j}|e_i\rg = \mathcal{R}_{ji}\,.
\end{equation}
With this identification  the relations \eqref{PMAPj} hold under the form,
\begin{equation}
	\label{PMAP1phi}
	\mathbb{P}_{\mathfrak{o}_1}M_f\mathbb{P}_{\mathfrak{o}_1}= U_{\mathfrak{o}_1}A^{(1)}_{f}U^{-1}_{\mathfrak{o}_1} \,, \quad A^{(1)}_{f}=  \int_0^{2\pi}  f(\phi) \,|\phi\rg_{e_1}{}_{e_1}\lg \phi|\,\frac{\ud \phi}{\pi}\,.
\end{equation}
\begin{equation}
	\label{PMAP2phi}
	\mathbb{P}_{\mathfrak{o}_2}M_f\mathbb{P}_{\mathfrak{o}_2}= U_{\mathfrak{o}_2}A^{(2)}_{f}U^{-1}_{\mathfrak{o}_2} \,, \quad A^{(2)}_{f}=  \int_0^{2\pi}  f(\phi) \,|\phi\rg_{e_2}{}_{e_2}\lg \phi|\,\frac{\ud \phi}{\pi}\,.
\end{equation}
With the identification $|e_1\rg = |\phi_0\rg$, $|e_2\rg = |\phi_0+ \pi/2\rg$, we finally have
\begin{equation}
	\label{Afrhophi}
	A^{(\rho_{r,\phi_0})}_{f}= \frac{1+r}{2}  A^{(\phi_0)}_{f} +\frac{1-r}{2}  A^{(\phi_0 +\pi/2)}_{f}\,,
\end{equation}
and we retrieve the same form as the density matrix \eqref{specrhor}.

\section{An example of  quantum measurement of orientations: polarization of light and Stokes parameters}
\label{quantmeas}

We now wish to describe the relationship between POVMs and light polarization within the context of integral quantization, which was briefly presented in \cite{Beneduci:2021iyv}. Let us first introduce the notion of polarization, and explain how it can be related to density matrices. As was pointed out in the preface of the textbooks \cite{Collett1993,Goldstein2003}, polarization is a fundamental characteristic of the transverse wave that is light. 
There exists a well-known  physical interpretation of $2\times2$ real or complex  density matrices in terms of the so-called Stokes parameters for the polarization of light, see for instance \cite{McMaster1954} and references therein. The significance and measurements of such parameters is clearly explained in \cite{schaefer2007measuring}. We follow here the presentation given by Landau and Lifshitz in Section 50 of \cite{Landau1975}, where is introduced the polarization ``tensor'' for a quasi-monochromatic plane wave propagating along the $z$-axis with the following notations involving the three Pauli matrices:
\begin{equation}
\label{poltensor}
\left(\rho_{\alpha \beta}\right) = \frac{1}{2}\begin{pmatrix}
1+ \xi_3      & \xi_1 -\ii \xi_2   \\
  \xi_1 + \ii \xi_2     &  1 - \xi_3
\end{pmatrix} = \frac{1}{2}\left(\bu+ \sum_{i=1}^3\xi_i \sigma_i\right)\,,
\end{equation}   
with the Stokes parameters
\begin{equation}
\label{stokland}
\xi_1 = r\sin2\phi\,,\quad \xi_2 = A\, , \quad \xi_3 =r \cos2\phi\,.
\end{equation}
The labels $\alpha$, $\beta$ run over the coordinates $(x,y)$ in the plane orthogonal to the $z$-axis. In \eqref{stokland} the parameter $0\leq r\leq 1$ characterizes the degree of maximal linear polarization whilst $-1\leq A\leq 1$ characterizes the circular polarization of the beam. \begin{figure}[h]
\begin{center}
\setlength{\unitlength}{0.1cm} 
\begin{picture}(60,60)
\put(10,10){\line(1,0){50}}
\put(10,10){\line(-1,0){30}}
\put(10,10){\vector(1,0){15}}
\put(10,10){\vector(0,1){15}}
\put(25, 7){\makebox(0,0){$\widehat{\boldsymbol{k}}$}} 
\put(8, 25){\makebox(0,0){$\widehat{\boldsymbol{\jmath}}$}} 
\put(22, 18){\makebox(0,0){$\widehat{\boldsymbol{\imath}}$}} 
\put(10,10){\line(0,1){40}} 
\put(10,10){\line(0,-1){10}} 
\put(10,10){\line(1,1){30}}
\put(10,10){\line(-1,-1){10}} 
\put(10,10){\vector(1,1){12}} 
%\multiput(10,10)(1,3){30}{\line(1,3){0.5}} 
%\multiput(10,10)(3,1){30}{\line(3,1){0.5}} 
%\put(10,10){\vector(1,3){23}} 
%\put(10,10){\vector(3,1){80}} 
%\put(10, 7.5){\makebox(0,0){$M$}} 
\put(10,10){\makebox(0,0){$\bullet$}}
%\put(34, 7.5){\makebox(0,0){$M$}} 
%\put(35, 10){\makebox(0,0){$\bullet$}}
%\put(35,10){\line(1,3){23}} 
%\put(60, 82){\makebox(0,0){$M^{\prime}$}}
%\put(60, 7.5){\makebox(0,0){$B$}}
%\put(60, 10){\makebox(0,0){$\bullet$}} 
\put(10,10){\vector(-2,1){25}} 
\put(-12,26.5){\makebox(0,0){$\mathfrak{Re}\left(\overrightarrow{\mathcal{E}}\right)$}}
%\put(60, 4){\makebox(0,0){Lightning}}
\put(07, 50){\makebox(0,0){$y$}} 
\put(60, 6.5){\makebox(0,0){$z$}} 
\put(40, 37){\makebox(0,0){$x$}} 
%\put(30, 79){\makebox(0,0){$t^{\prime}$}} 
%\put(88, 40){\makebox(0,0){$x^{\prime}$}} 
%\put(78, 73){\makebox(0,0){Signal from left end}} 
%\put(-1, 73){\makebox(0,0){Signal from right end}} 
%\put(41.3, 28.4){\makebox(0,0){$\bullet$}} 
%\put(47, 29){\makebox(0,0){$E_{\mathrm{right}}$}} 
%\put(47.8, 47.8){\makebox(0,0){$\bullet$}} 
%\put(52, 46.3){\makebox(0,0){$E_{\mathrm{left}}$}} 
\thicklines 
\end{picture}
\end{center}
\caption{A quasi-monochromatic plane wave propagates to the right along the $z$-axis. The electric field, as the real part of its complex description,   lies in the $x-y$ plane.}
\label{proplight}
\end{figure}
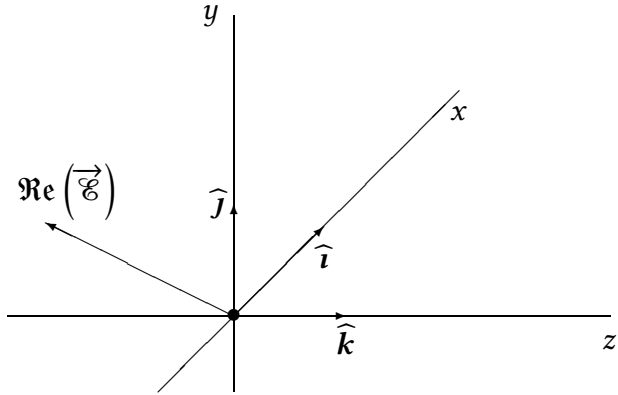
Let us explain the physical origin of the polarization tensor  \eqref{poltensor}. In its complex description, the electric field reads
\begin{equation}
\label{celfield}
\overrightarrow{\mathcal{E}}(t)= \overrightarrow{\mathcal{E}_0}(t)\,e^{\ii \omega t}= \mathcal{E}_x\,\bi + \mathcal{E}_y\,\bj = \left(\mathcal{E}_{\alpha }\right)\,, 
\end{equation}
where $\overrightarrow{\mathcal{E}_0}(t)$ has slow temporal variation and $\omega$ is the mean frequency. $\overrightarrow{\mathcal{E}_0}$ determines the polarization of light. The latter is measured through Nicol prisms, or other devices, by measuring the intensity of the light yielded by mean values of quadratic expressions of field components, which are proportional to $\mathcal{E}_{\alpha}\mathcal{E}_{\beta}$, $\mathcal{E}_{\alpha}\mathcal{E}^{\ast}_{\beta}$ and their respective complex conjugates. Now,  $\mathcal{E}_{\alpha}\mathcal{E}_{\beta}=\mathcal{E}_{0\alpha}\mathcal{E}_{0\beta} e^{2\ii \omega t}$ and  $\mathcal{E}^{\ast}_{\alpha}\mathcal{E}^{\ast}_{\beta}=\mathcal{E}^{\ast}_{0\alpha}\mathcal{E}^{\ast}_{0\beta} e^{-2\ii \omega t}$  have rapidly oscillating factors and so have null temporal average $\lg\cdot\rg_t$. It results that the properties of a partially polarized light are completely described by the tensor   
\begin{equation}
\label{Jalbet}
J_{\alpha\beta}:=\left\lg \mathcal{E}_{0\alpha}\mathcal{E}^{\ast}_{0\beta}\right\rg_t\,.
\end{equation}
The quantity $J=\sum_{\alpha}J_{\alpha\alpha}= \left\lg \vert \mathcal{E}_{0x}\vert^2\right\rg_t + \left\lg \vert \mathcal{E}_{0y}\vert^2\right\rg_t$ determines the intensity of the wave, obtained from the measurement of the energy flux transported by the wave. Since this quantity does not concern the properties of polarization of the wave, one has eventually to deal with the normalized tensor \eqref{poltensor}, \textit{i.e.}, 
\begin{equation}
\label{Jrho}
\rho_{\alpha \beta} = \frac{J_{\alpha\beta}}{J}\,.
\end{equation}
The light is said completely polarized when the complex amplitude $\overrightarrow{\mathcal{E}}_0$ is time-independent, and so is equal to its time average. Then, the polarization tensor factorizes as
\begin{equation}
\label{totpol}
\left(\rho_{\alpha \beta}\right) = \frac{1}{J}\begin{pmatrix}
 \vert \mathcal{E}_{0x} \vert^2    & \mathcal{E}_{0x}\mathcal{E}^{\ast}_{0y}   \\
  \mathcal{E}^{\ast}_{0x}\mathcal{E}_{0y}    &   \vert \mathcal{E}_{0y} \vert^2
\end{pmatrix}= \begin{pmatrix}
    \mathcal{E}_{0x} /\sqrt{J}      \\
    \mathcal{E}_{0y}  /\sqrt{J}    
\end{pmatrix}\begin{pmatrix}
  \mathcal{E}^{\ast}_{0x} /\sqrt{J}     & \mathcal{E}^{\ast}_{0y} /\sqrt{J}    \\  
\end{pmatrix} \,,
\end{equation}   
 \textit{i.e.}, is the orthogonal projector along $\overrightarrow{\mathcal{E}}_0$, and, in quantum terms, a pure state.  This case should be put in regard to the other extreme case, namely the non-polarized or natural light, for which all directions in the $x-y$ plane are equivalent:
 \begin{equation}
\label{non-pol}
\left(\rho_{\alpha \beta}\right) = \frac{1}{2}\,\delta_{\alpha\beta}\,.
\end{equation}
In the general case and with $\xi_i$ coordinates,  $\det\left(\rho_{\alpha \beta}\right) = \frac{1}{4}\left(1-\xi_1^2-\xi_2^2 -\xi_3^2\right)\equiv  \frac{1}{4}\left(1-P^2\right)$, where $0\leq P\leq 1$ is called  the degree of polarization, from $P=0$ (random polarization) to $P=1$ (total polarization). Another extreme case holds  with circular polarization. Then $\overrightarrow{\mathcal{E}_0}$ is constant and $\mathcal{E}_{0y}= \pm \ii \mathcal{E}_{0x}$, which gives $\left(\rho_{\alpha \beta}\right)= (1/2)(\bu\pm\ii \sigma_2)$. As a consequence, the parameter $-1\leq A\leq 1$ in \eqref{stokland} is interpreted as the degree of circular polarization, with $A=1$ (resp. $A=-1$) for right (resp. left) circular polarization, and  $A=0$ for linear polarization. 

This distinction between linear and circular polarizations corresponds to the decomposition of \eqref{poltensor} into its symmetric and antisymmetric parts:
\begin{equation}
\label{rolroh}
\left(\rho_{\alpha \beta}\right) = \rho_{r,\phi} + \frac{A}{2}\sigma_2 = \frac{1+r}{2} E_{\phi} +  \frac{1-r}{2} E_{\phi + \pi/2} + \frac{A}{2}\sigma_2 \,,
\end{equation}
where we have introduced our own notations in \eqref{specrho} and  \eqref{standrhomain}. Hence, the light can be viewed as the superposition of two waves which are incoherent and elliptically polarized, with similar and orthogonal polarization ellipses. 
Within the  context of the present paper, circular polarization is not considered. We then put $A=0$, ultimately reducing the polarization tensor to our density matrix.

Let us now describe the interaction polarizer-partially linear polarized light as an example of  simple quantum measurements. Here, we follow and generalize  the example presented by Peres in \cite{Peres1990}. Two planes and their tensor product are under consideration. The first one is the Hilbert space on which act the states $\rho^M_{s,\theta}$ of the polarizer viewed as an  orientation \textit{pointer}. We note that the action of the generator of rotations   $\tau_2=-\ii\sigma_2$ on these states corresponds to a $\pi/2$ rotation :
\begin{equation}
\label{tau2act}
\tau_2\rho^M_{s,\theta}\tau_2^{-1} = -\tau_2\rho^M_{s,\theta}\tau_2 = \rho^M_{s,\theta +\pi/2}\,.
\end{equation}
The second plane is  the Hilbert space on which act the partially linearized polarization states $\rho^L_{r,\phi}$ of the plane wave crossing the polarizer. As was explained above, the spectral decomposition of the latter corresponds to the incoherent superposition of two completely linearly polarized waves 
\begin{equation}
\label{pointer}
\rho^L_{r,\phi}=\frac{1+r}{2} \, E_{\phi} + \frac{1-r}{2} \, E_{\phi + \pi/2}\, . 
\end{equation}
The pointer is designed to detect the  orientation in the plane determined by the angle $\phi$.
The interaction pointer-system generating a measurement whose time duration is the interval $I_M=(t_M-\eta, t_M +\eta)$ centred at $t_M$ is described by the  (pseudo-) Hamiltonian operator 
\begin{equation}
\label{intham}
\widetilde{H}_{\mathrm{int}}(t)= \, g^{\eta}_M(t)\tau_2 \otimes \rho^L_{r,\phi} \,, 
\end{equation}
where $g^{\eta}_M$ is a Dirac sequence with support in $I_M$, \textit{i.e.},  $$\lim_{\eta\to 0}\int_{-\infty}^{+\infty}\ud t \, f(t) \,g^{\eta}_M(t)= f(t_M)\,.$$ 
The operator \eqref{intham} is the tensor product of an antisymmetric operator (\textit{i.e.}, or pseudo-Hamiltonian, see \cite{bergeron2019orientations}) for the pointer with an operator for the system which is symmetric (\textit{i.e.} Hamiltonian). Nevertheless,  the operator $U(t,t_0)$ defined  for $t_0< t_M-\eta$ as  
\begin{equation}
\label{nevop}
U(t,t_0)= \exp\left[\int_{t_0}^{t}\ud t^{\prime}\, g^{\eta}_M(t^{\prime})\, \tau_2 \otimes \rho^L_{r,\phi}\right] = \exp\left[ G_M^\eta(t)\, \tau_2 \otimes \rho^L_{r,\phi}\right]\, , 
\end{equation} 
with $G_M^\eta(t)=\int_{t_0}^{t}\ud t^{\prime} \, \,g^{\eta}_M(t^{\prime})$, is an evolution operator. Hence, as soon as $t>t_M+\eta$, $G_M^\eta(t)=1$, and from the general formula involving an orthogonal projector $P$,  
\begin{equation}
\label{expP}
\exp(\theta \tau_2 \otimes P) = \mathcal{R}(\theta) \otimes  P + \bu\otimes (\bu-P)\, , 
\end{equation}
we obtain
\begin{equation}
\label{nevop1t}
 U(t,t_0)= \mathcal{R}\left(G_M^\eta(t)\,\frac{1+r}{2}\right) \otimes  E_{\phi} + \mathcal{R}\left(G_M^\eta(t)\,\frac{1-r}{2}\right)  \otimes  E_{\phi +\pi/2}\, .
\end{equation}
For $t_0<  t_M-\eta$ and $t> t_M+\eta$, we finally obtain
\begin{equation}
\label{nevop1}
 U(t,t_0)= \mathcal{R}\left(\frac{1+r}{2}\right) \otimes  E_{\phi} + \mathcal{R}\left(\frac{1-r}{2}\right)  \otimes  E_{\phi +\pi/2}\, . 
\end{equation}
One easily checks that  $U(t,t_0) U(t,t_0)^{\dag} = U(t,t_0)^{\dag} U(t,t_0)  = \bu\otimes \bu$. 
After having prepared the polarizer in the state $\rho^M_{s_0,\theta_0}$,  the  action of \eqref{nevop1} on the initial state $\rho^M_{s_0,\theta_0} \otimes \rho^L_{r_0,\phi_0}$ reads for $t> t_M+\eta$
\begin{align}
\label{actinUtt0}
\nonumber U(t,t_0)\,\rho^M_{s_0,\theta_0} \otimes \rho^L_{r_0,\phi_0}\,U(t,t_0)^{\dag}=& \rho^M_{s_0,\theta_0+\frac{1+r}{2}} \otimes \frac{1+r_0\cos2(\phi-\phi_0)}{2}\,E_{\phi}\\
\nonumber  & + \rho^M_{s_0,\theta_0+\frac{1-r}{2}} \otimes \frac{1-r_0\cos2(\phi-\phi_0)}{2}\,E_{\phi+\pi/2}\\
\nonumber  &+ \frac{1}{4}\left(\mathcal{R}(r) +s_0\sigma_{2\theta_0 +1}\right)\otimes r_0\sin2(\phi-\phi_0)\,E_{\phi}\tau_2\\
 &- \frac{1}{4}\left(\mathcal{R}(-r) +s_0\sigma_{2\theta_0 +1}\right)\otimes r_0\sin2(\phi-\phi_0)\,\tau_2 E_{\phi}\, . 
\end{align}
As expected from the standard  theory of quantum measurement,  this formula indicates that the probability for the pointer to rotate by $\frac{1+r}{2}$, corresponding to the polarization  along the orientation  $\phi$, is 
\begin{equation}
\label{ pointerphi}
\mathrm{Tr}\left[\left(U(t,t_0)\,\rho^M_{s_0,\theta_0} \otimes \rho^L_{r_0,\phi_0}\,U(t,t_0)^{\dag}\right)\left(\bu\otimes E_\phi\right)\right]= \frac{1+r_0\cos2(\phi-\phi_0)}{2}\,,
\end{equation} 
  whereas it is 
\begin{equation}
\label{ pointerphiorth}
\mathrm{Tr}\left[\left(U(t,t_0)\,\rho^M_{s_0,\theta_0} \otimes \rho^L_{r_0,\phi_0}\,U(t,t_0)^{\dag}\right) \left(\bu\otimes E_{\phi+\pi/2}\right)\right]= \frac{1-r_0\cos2(\phi-\phi_0)}{2}\,,
\end{equation}  
for the perpendicular orientation $\phi + \pi/2$ and the pointer rotation by $\frac{1-r}{2}$. For the completely linear polarization of the light, i.e. $r_0=1$, we recover the familiar Malus laws, $\cos^2 (\phi-\phi_0)$ and $\sin^2 (\phi-\phi_0)$ respectively.

\section{Compatibility of POVMs}\label{Compatibility}

We now argue that the physical situation where two measurements are made one after another is described by a dichotomic POVM. We then recall the necessary conditions for the compatibility of two dichotomic POVMs, and what form these conditions take in the Euclidean plane.

\subsection{POVMs arising from the measurement of polarization}\label{sequential}

\noindent
It is well known \cite{davies1976quantum} that the statistics of the sequential measurement of two incompatible observables is described by a POVM. Let us analyze the case of the sequential measurement of two incompatible polarization observables.
Suppose an incident light ray is polarized according to the state $\rho$ and is sent toward a first vertical polarizer ($\pi/2$) and subsequently toward a second oblique ($\pi/4$) polarizer. The light ray is then detected by a light detector. Suppose the state of the detector is $D=1$ if it detects a light ray and $D=0$ if it does not. What is the probability that the light ray triggers the detector (in which case $D=1$ and the polarization of the transmitted light ray is $\pi/4$)?

We first note that the probability that the light ray exits the first polarizer is $Tr \left(\rho E_{\frac{\pi}{2}} \right)$ where $\rho$ is an initial state and $E_{\frac{\pi}{2}}=|\frac{\pi}{2}\rangle\langle\frac{\pi}{2}|$. If the light ray has passed the first polarizer, its state is
$$\rho\rq{}=\frac{E_{\frac{\pi}{2}} \rho E_{\frac{\pi}{2}}}{Tr \left(E_{\frac{\pi}{2}}\rho \right)} \;.$$
Then, the probability that it passes the second polarizer is $Tr \left(\rho\rq{} E_{\frac{\pi}{4}}\right)$. Therefore, the probability that the state of the detector is $D=1$ at the end of the measurement process (the light ray is transmitted) is 
\begin{align}
p\left(D=1\right)=Tr\left(\rho E_{\frac{\pi}{2}}\right)Tr\left(\rho\rq{} E_{\frac{\pi}{4}}\right)=Tr\left(\rho E_{\frac{\pi}{2}} E_{\frac{\pi}{4}} E_{\frac{\pi}{2}}\right)=\frac{1}{2}Tr\left(\rho E_{\frac{\pi}{2}}\right).
\end{align}
Note that $E_{\frac{\pi}{2}} E_{\frac{\pi}{4}} E_{\frac{\pi}{2}}$ is an effect, \textit{i.e.} a symmetric operator $E$ such that $0\leq E\leq \bu$, but not a projection operator (it is a multiple of a projection operator).
Now, the probability that the light ray passes the first polarizer but does not pass the second is 
\begin{align}
Tr\left(\rho E_{\frac{\pi}{2}}\right)Tr\left[\rho\rq{} \left(\bu-E_{\frac{\pi}{4}}\right)\right]=Tr\left[\rho E_{\frac{\pi}{2}} \left(\bu-E_{\frac{\pi}{4}}\right) E_{\frac{\pi}{2}}\right]=\frac{1}{2}Tr\left(\rho E_{\frac{\pi}{2}}\right). 
\end{align}
The probability that, at the end of the measurement process, the detector is in the state $D=0$ is
\begin{align}
p(D=0)=Tr(\rho E_0)+\frac{1}{2}Tr(\rho E_{\frac{\pi}{2}})=Tr\left[\rho \left(E_0+\frac{1}{2}E_{\frac{\pi}{2}}\right)\right].
\end{align}
Note that $p(D=1)+p\big(D=0\big)=1$. Again, $E_0+\frac{1}{2}E_{\frac{\pi}{2}}$ is an effect but not a projection operator. The measurement process is then described by the dichotomic POVM 
\begin{align}
F=\left\lbrace\frac{1}{2}E_{\frac{\pi}{2}} \;, \; E_0+\frac{1}{2}E_{\frac{\pi}{2}}\right\rbrace \;,
\end{align}
which is not the spectral measure corresponding to a self-adjoint operator. 

If we analyze the sequential measurement where the light ray first passes the oblique polarizer and then the vertical one, we obtain that the overall measurement process is characterized by the POVM $\left\lbrace\frac{1}{2}E_{\frac{\pi}{4}}, \bu-\frac{1}{2}E_{\frac{\pi}{4}}\right\rbrace$ which is different from the previous one describing the measurement process with the light ray passing first through the vertical polarizer. The two measurement procedures give different results. That is explained by the incompatibility of the measurement of the vertical and the oblique polarization and is analogous to the Heisenberg microscope experiment where a sequential measure of the position and momentum is realized.

\subsection{Dichotomic POVMs}

A dichotomic POVM $F$ is a pair $F=\{A, \bu-A\}$ with $A$  an effect.
We consider two dichotomic POVMs $F$ and $F'$, and look for necessary and sufficient conditions for their compatibility. We show that every dichotomic observable $F$ is the fuzzy version of a sharp observable $E$; $F$ can also be interpreted as an approximation of $E$. Then, inspired by \cite{busch2016quantum}, we answer the following question: given two incompatible sharp observables $E$ and $E'$, is it possible to define two compatible fuzzy versions $F$ and $F'$ of $E$ and $E'$ respectively whose joint measurement can be interpreted as the approximate joint measurement of $E$ and $E'$? First we recall the main definitions and theorems concerning the compatibility of POVMs that we adapt to the case of a real, bidimensional Hilbert space.
Two dichotomic POVMs $F_1=\{A_1,\bu-A_1\}$ and $F_2=\{A_2,\bu-A_2\}$ are compatible if there is a third POVM $F=\{G_{11},G_{00},G_{10},G_{01}\}$ of which $F_1$ and $F_2$ are the marginals, \textit{i.e.}, 
\begin{align}\label{compatible}
&A_1=G_{11}+G_{10},\,\, A_2=G_{11}+G_{01},\\
&\bu-A_1=G_{00}+G_{01},\,\, I-A_2=G_{00}+G_{10}\notag \;.
\end{align}
A necessary and sufficient condition for the compatibility \cite{busch2007approximate} of $F_1$ and $F_2$ is the existence of a positive operator $G_{11}$ such that 
\begin{equation}\label{CNS}
\begin{split}
&G_{11} \leq A_1, A_2 \;, \\
	&A_1+A_2-G_{11} \leq \bu \;.
\end{split}
\end{equation}

Indeed, if $G_{11}$ exists, we can define the effects 
\begin{align}
	G_{00}&=\bu-A_1-A_2+G_{11}\geq 0 \;, \\
	G_{10}&=A_1-G_{11}\geq 0 \;, \\
	G_{01}&=A_2-G_{11}\geq 0  \;,
\end{align}
which satisfy \eqref{compatible}.

Now we translate the above compatibility conditions in the present context where an observable is represented by the dichotomic observable corresponding to a symmetric operator of the kind defined below 
\begin{equation}\label{random}
A(\alpha,\phi,r)=
 \begin{pmatrix}
  \frac{\alpha}{2}  + \frac{r}{2}\cos2\phi  &   \frac{r}{2}\sin2\phi  \\
\frac{r}{2}\sin2\phi    &   \frac{\alpha}{2}  - \frac{r}{2}\cos2\phi
\end{pmatrix}= \alpha \rho_{\frac{r}{\alpha},\phi}\, . 
\end{equation}

\noindent
with eigenvalues $\frac{\alpha}{2}\pm \frac{r}{2}$. The condition $0\leq A\leq I$ then becomes $r\leq\alpha\leq 2-r$ (so that $0\leq\alpha\leq2$ and $0\leq r\leq1$) which ensures that $A\left(\alpha,\phi,r\right)$ is an effect. Note that there are two meanings we can give to matrices of the form \eqref{random}: if $\alpha=r=1$, they are density operators which can be used to represent pure states or sharp observables;  in the case $\alpha=1$, $r\neq1$, they can represent mixed states or fuzzy (unsharp) observables; in the general case $\alpha,r\neq1$ they represent fuzzy unbiased observables (see below). We recall that in the context of the linear polarisation of the light as described by \eqref{Jrho}, the parameter $\alpha$ may be interpreted as the intensity of the light while $r=\xi_1^2+\xi_3^2=P$ provides a measure of the degree of polarization. That is relevant in order to give an interpretation of the fuzzy observables as  approximations of the sharp observables and of the joint measurability of fuzzy observables as the approximate joint measurability of incompatible sharp polarization observables (see below).

The dichotomic observable corresponding to $A(\alpha,\phi,r)$ is $\big\{F_+=A(\alpha,\phi,r),\,F_-=\bu-A(\alpha,\phi,r)\big\}$ and can be interpreted as the randomization of the dichotomic PVM $\big\{E_+=A(1,\phi,1),\,E_-=\bu-A(1,\phi,1)\big\}$ through the Markov kernel $\mu(+,+)=\frac{\alpha+r}{2}$, $\mu(+,-)=\frac{\alpha-r}{2}$, $\mu(-,+)=1-\mu(+,+)$, $\mu(-,-)=1-\mu(+,-)$. Indeed, $F_+=\mu(+,+)E_++\mu(+,-)E_-$ and $F_-=\mu(-,+)E_++\mu(-,-)E_-$.

As an example, we can consider the POVMs $\left\lbrace\frac{1}{2}E_{\frac{\pi}{2}},\bu-\frac{1}{2}E_{\frac{\pi}{2}}\right\rbrace$ and $\left\lbrace\frac{1}{2}E_{\frac{\pi}{4}},\bu-\frac{1}{2}E_{\frac{\pi}{4}}\right\rbrace$ that we found in section \ref{sequential} which represent sequential measurements. In particular, we have that $\frac{1}{2}E_{\frac{\pi}{2}}=A\left(\frac{1}{2},\frac{\pi}{2},\frac{1}{2}\right)$, $\frac{1}{2}E_{\frac{\pi}{4}}=A\left(\frac{1}{2},\frac{\pi}{4},\frac{1}{2}\right)$. They are not compatible as can be seen by showing that there is no symmetric positive operator satisfying conditions (\ref{CNS}).  Note also that  $\left\lbrace \frac{1}{2}E_{\frac{\pi}{2}},\bu-\frac{1}{2}E_{\frac{\pi}{2}} \right\rbrace$ is the randomization of the PVM  $\left\lbrace E_+=A\left(1,\frac{\pi}{2},1\right),\,E_-=\bu-A\left(1,\frac{\pi}{2},1\right) \right\rbrace$ through the Markov kernels $\mu(+,+)=\frac{1}{2}$, $\mu(+,-)=0$, $\mu(-,+)=\frac{1}{2}$, $\mu(-,-)=1$. The same holds for the POVM $\left\lbrace\frac{1}{2}E_{\frac{\pi}{4}},\bu-\frac{1}{2}E_{\frac{\pi}{4}}\right\rbrace$ which is the randomization of the PVM $\left\lbrace E_+=A\left(1,\frac{\pi}{4},1\right),\,E_-=\bu-A\left(1,\frac{\pi}{4},1\right)\right\rbrace$ through the same Markov kernel.

Let  $$F_1=\{A(\alpha_1,\phi_1,r_1),\bu-A(\alpha_1,\phi_1,r_1)\}$$ and $$F_2=\{A(\alpha_2,\phi_2,r_2), \bu-A(\alpha_2,\phi_2,r_2)\}$$ be two dichotomic POVMs. We look for conditions for the compatibility of $F_1$ and $F_2$. An analogous problem has been dealt with in Ref. \cite{busch2007approximate,busch2016quantum}. The present case can be considered as the real version of the qubit case analyzed in \cite{busch2016quantum}, chapter 14. If circular polarization were included, the relevant Hilbert space would be $\mathbb{C}^2$ and the formalism would coincide with the one in \cite{busch2016quantum}. 

First, we note that to each symmetric operator $A(\alpha,\phi,r)$ there corresponds a vector $v=(r,\phi)$ in the unit closed disc. We then have
\begin{align*}
	A_1&=A(\alpha_1,\phi_1,r_1)\mapsto v_1=(r_1,\phi_1) \;, \\
	A_2&=A(\alpha_2,\phi_2,r_2)\mapsto v_2=(r_2,\phi_2) \;, \\
	G_{11}&=A(\alpha,\phi,r)\mapsto v=(r,\phi) \;.
\end{align*}
Moreover, using \eqref{compatible}
\begin{align*}
	A_1-G_{11}&=A(\alpha_1,\phi_1,r_1)-A(\alpha,\phi,r)=G_{10}=A(\alpha_1-\alpha,r_{10},\phi_{10}) \;, \\
	A_2-G_{11}&=A(\alpha_2,\phi_2,r_2)-A(\alpha,\phi,r)=G_{01}=A(\alpha_2-\alpha,r_{01},\phi_{01}) \;, \\
	\bu-A_1-A_2+G_{11}&=\bu-A(\alpha_1,\phi_1,r_1)-A(\alpha_2,\phi_2,r_2)+A(\alpha,\phi,r)\\ &=G_{00}=A(1-\alpha_1-\alpha_2+\alpha, r_{00},\phi_{00}) \;,
\end{align*}
and consequently, we can define the vectors
\begin{align*}
	v_{10}&=(r_{10},\phi_{10})=v_1-v \;, \\ v_{01}&=(r_{01},\phi_{01})=v_2-v \;,\\ v_{00}&=(r_{00},\phi_{00})=v-v_1-v_2 \;.
\end{align*}
This can be seen by a simple application of Carnot\rq{}s theorem for triangles. Therefore, the respective eigenvalues imply that the compatibility between POVMs is achieved if
\begin{align}\label{conditions}
	\|v\|&=r\leq\alpha  &&(\text {positivity of}\,\, G_{11}) \;, \\
	\|v_1-v\|&=r_{10}\leq\alpha_1-\alpha  &&(\text{positivity of}\,\, G_{10}) \;, \notag\\
	\|v_2-v\|&=r_{01}\leq\alpha_2-\alpha &&(\text{positivity of}\,\, G_{01}) \;, \notag\\
	\|v-v_1-v_2\|&=r_{00}\leq 2-\alpha_1-\alpha_2+\alpha &&(A_1+A_2-G_{11}\leq \bu)\;.\notag 
\end{align}
We then have
\begin{align}\label{necessary}
\|v_1+v_2\|+\|v_1-v_2\|\leq\alpha+(2-\alpha_1-\alpha_2+\alpha)+(\alpha_1-\alpha)+(\alpha_2-\alpha)= 2
\end{align}
which is a necessary condition for the compatibility. That is exactly the same condition obtained in \cite{busch2007approximate} for the case of a two-dimensional complex Hilbert space of which the real case we are considering is a particular case.  However, the condition is not sufficient, as can be seen by noting that the two dichotomic observables introduced in the previous subsection, \textit{i.e.}, $A\left(\frac{1}{2},\frac{\pi}{2},\frac{1}{2}\right)$ and $A\left(\frac{1}{2},\frac{\pi}{4},\frac{1}{2}\right)$, satisfy the necessary condition but are not compatible as we remarked above. Note also that in the case of unbiased dichotomic observables $A(1,\phi_1,r_1)$, $A(1,\phi_2,r_2)$, the condition is also sufficient. Indeed, one can choose $v=\frac{v_1+v_2}{2}$ and $\alpha=1-\frac{|v_1-v_2|}{2}$ which, thanks to condition (\ref{necessary}), ensures that the inequalities (\ref{conditions}) are satisfied. 

In order to give an interpretation of both the fuzzy observable $\{A(\alpha,\phi,r), \bu-A(\alpha,\phi,r)\}$ and the necessary condition (\ref{necessary}), we can resort again to the Stokes parameters. In section \ref{quantmeas}, we have shown that the quantity $P=\xi_1^2+\xi_3^2=r$ provides a measure of the degree of linear polarization; with $P$ varying between $1$ (total polarization) and $0$ (random polarization). In the present framework, the same quantity is interpreted as a measure of the fuzziness of the polarization measurement. For example, a measurement of polarization represented by the PVM  $E=\big\{E_+=A(1,\phi,1),\,E_-=\bu-A(1,\phi,1)\big\}$ (sharp observable) corresponds to a sharp measurement of polarization in the direction defined by $\phi$ while a measurement represented by the POVM $F=\big\{F_+=A(1,\phi,r),\,F_-=\bu-A(1,\phi,r)\big\}$, $r<1$ (unsharp observable) corresponds to a fuzzy measurement; the fuzziness depending on $r$ which measures the degree with which $F$ approximates $E$. The POVM $\hat{F}=\big\{F_+=A(\alpha,\phi,r),\,F_-=\bu-A(\alpha,\phi,r)\big\}$ is a fuzzy version of $E$ as well and the degree of approximation depends both on $\alpha$ and $r$. If two sharp observables $E=\big\{E_+=A(1,\phi,1),\,E_-=\bu-A(1,\phi,1)\big\}$ and $E'=\big\{E_+=A(1,\phi',1),\,E_-=\bu-A(1,\phi',1)\big\}$ are not compatible, they could be approximated by two compatible fuzzy versions $F=\big\{F_+=A(1,\phi,r),\,F_-=\bu-A(1,\phi,r)\big\}$ and $F'=\big\{F_+=A(1,\phi',r'),\,F_-=\bu-A(1,\phi',r')\big\}$. Equation (\ref{necessary}) gives a necessary and sufficient condition for the joint measurability of $F$ and $F'$ to be connected to the degree of fuzziness one has to add in the measurement process; the latter being measured by $r$ and $r'$. In the case we use  fuzzy observables  $F=\big\{F_+=A(\alpha,\phi,r),\,F_-=\bu-A(\alpha,\phi,r)\big\}$ and $F'=\big\{F_+=A(\alpha',\phi',r'),\,F_-=\bu-A(\alpha',\phi',r')\big\}$, with $\alpha\neq 1$ and/or $\alpha'\neq 1$, to approximate $E$ and $E'$, condition (\ref{necessary}) is only necessary.

\section{Conclusions}

We considered in this work the role of real Positive Operator-Valued Measures as quantizers of a physical system as well as their status as quantum observables in finite dimension. We focused in particular in dimension 2, which is the lowest one for which we can develop a non-trivial quantum formalism. The starting point, after having recalled the definition of POVMs, was to lay down some necessary and sufficient conditions for two POVMs to be compatible. We then emphasized that POVMs, besides their role as observables, can also provide a quantization scheme in the framework of integral quantization. In particular, a first result emerged when considering Toeplitz quantization as an integral quantization, where we showed that Naimark theorem is valid in our case.

We continued the study by restraining ourselves to integral quantizations of functions on the unit circle and their actions on real two-dimensional systems, and where we supposed non-negative POVMs, \textit{i.e.}  density matrices. In this case, the quantization map gives a non-commutative version of $\mathbb{R}^3$ that can be identified as a Fourier subspace, and the density matrix \eqref{standrhomain} is particularly simple. Quantizing a function on this subspace gives the simple $2$x$2$ matrix representation \eqref{qtfrhor}, which allows for a simple non-commutative version of $\mathbb{R}^3$. The most general functions living in this Fourier space were found in \eqref{frho2}. They were subsequently used to prove that a mixed state can be represented as a continuous superposition of mixed state, which is a typical quantum mechanical behavior. Applying Naimark theorem made possible the identification between the quantum operator of a POVM and its density matrix.

A physical illustration of the previous results was made by using light polarization. Remembering the construction of the polarization tensor \eqref{rolroh}, we showed that it reduces to the general form of our two-dimensional density matrices when circular polarization is neglected. Then, the interaction between a polarizer and the system was described during a quantum measurement, where we showed that the probability to obtain a given polarization is given by the evolution equation \eqref{actinUtt0}.

The case of sequential measurements was analysed at last. It was shown that two measurements in which a light ray goes first through an oblique polarizer before passing through a vertical polarizer is described by a dichotomic POVM, while a measurement in the reverse order is described by another dichotomic POVM, showing the incompatibility of the measurement procedures. We finally searched for the necessary condition \eqref{necessary} for the compatibility of two dichotomic POVMs in a real bidimensional Hilbert space. In the last section, we related the density matrix's parameters to the Stokes parameters defining a polarization tensor for linearly polarized light. It turns out that we can identify the degree of mixing of a density matrix to the fuzziness of a quantum observable. We concluded that compatibility conditions of two POVMs can be expressed in terms of Stokes parameters.

\subsection*{Acknowledgments}

JPG thanks the Brazilian Center for Research in Physics (CBPF) for its hospitality. EF thanks both the CBPF and the Helsinki Institute of Physics (HIP) for their hospitality. The present work has been realized in the framework of the activities of the INDAM (Istituto Nazionale di Alta Matematica).

\appendix

\section{Quantum orientations in $\R^n$}
 \label{QSON}
 \subsection{Pure and mixed states}
 We now generalize to $\R^n$ a part of the material developed at length in the previous sections. A vector in $\R^n$ is denoted by $\vv$ (Euclidean geometry) or $|\vv\rg$ (Dirac or Hilbertian notation). Thus a pure state is represented by an element of $\mathbb{S}^{n-1}/\mathbb{Z}_2$, and a mixed state by the general density matrix
\begin{equation}
\label{nrho}
\rho\equiv\rho_{\vet,\vph }=\frac{1}{n} \bu + \mathcal{R}(\vph )\,\sfD(\vet) \,
{}^{t}\mathcal{R}(\vph ) \, ,  
\end{equation} 
 with notations similar to \eqref{standrhomain}. The diagonal matrix $\bu/n \equiv \rho_{\mathrm{rm}}$ which appears in this expression  describes the full random mixing. 
The symbol $\sfD$ stands for the diagonal matrix
\begin{equation}
\label{sfDeta}
\sfD(\vet) =\begin{pmatrix} 
    \eta_1 &  \dots &0 \\
    \vdots & \ddots  &\vdots\\
    0 &   \dots   &  \eta_n 
    \end{pmatrix}\,,
\end{equation}
where the vector $\vet= {}^t(\eta_1,\eta_2, \cdots,\eta_n)$ lies in the simplex defined by the conditions
\begin{equation}
\label{simpleta}
\sum_{i=1}^n \eta_i= 0\, , \quad -\frac{1}{n}\leq \eta_i \leq 1 -\frac{1}{n}\, ,
\end{equation} 
\textit{i.e.}, the intersection of the hyperplane normal to the first diagonal and containing the origin with the unit hypercube shifted by the vector ${}^t(-1/n, -1/n, \dotsc, -1/n)$. The rotation matrix  $\mathcal{R}(\vph )\in \mathrm{SO}(n)$ and the symbol $\vph$ stands for the set of $n(n-1)/2$ (angular) parameters of the group SO$(n)$. 
Thus, the manifold of density matrices \eqref{nrho} is real with dimension equal to $n(n+2)/2$.

 The representation of the SO$(n)$ group composition law  is written here as
\begin{equation}
\label{repson}
\mathcal{R}(\vph)\mathcal{R}(\vph^{\prime})= \mathcal{R}(\vph\circ \vph^{\prime}) \;.
\end{equation}
We adopt the Euler parametrisation based on the descending sequence  $\mathrm{SO}(n)\supset\mathrm{SO}(n-1)\supset\cdots\supset\mathrm{SO}(2)$  \cite{vilenkin1978special}. First, for $1\leq k\leq n-1$, we denote $\mathcal{R}_{k}(\phi)$ the rotation by the angle $\phi$ in the $\left(x_k,x_{k+1}\right)$-plane. Next, we introduce the successive products of such elementary rotations:   
\begin{equation}
\label{krot}
\mathcal{R}^{(k)}\left(\vph^{(k)}\right)= \mathcal{R}_{1}(\phi_1^k)\,\mathcal{R}_{2}(\phi_2^k)\,\cdots \mathcal{R}_{k}(\phi_k^k)\, , \quad \vph^{(k)}= \left(\phi_1^k,\phi_2^k,\dotsc,\phi_k^k\right)\,. 
\end{equation}
For each $1\leq k \leq n-1$, the Euler angles $\phi_j^k$, $1\leq j\leq k$ vary as follows
\begin{equation}
\label{eulerdom}
\left\lbrace\begin{array}{cc}
    0\leq  \phi_1^k  < 2\pi&    \\
     0\leq  \phi_j^k  < \pi &   j\neq 1 
\end{array}\right. \, .
\end{equation}
Then any rotation  $\mathcal{R}(\vph)$ of the group SO$(n)$ can be written \cite{vilenkin1978special} as the product
\begin{equation}
\label{paramson}
\mathcal{R}(\vph)= \prod_{k=1}^{n-1}\mathcal{R}^{(k)}\left(\vph^{(k)}\right)\, , \quad \vph= \left( \vph^{(1)},\dotsc, \vph^{(n-1)}\right)\,. 
\end{equation}
The non-normalized invariant (Haar) measure on SO$(n)$ reads in terms of the  above $n(n-1)/2$ Euler angular variables:
\begin{equation}
\label{haarson}
\ud \vph= \prod_{k=1}^{n-1}\prod_{j=1}^{k}\,\sin^{j-1}\phi_j^k\,\ud\phi_j^k\,, 
\end{equation} 
from which we obtain the volume of SO$(n)$
\begin{equation}
\label{volson}
\mathrm{Vol} \left(\mathrm{SO}(n)\right)= \int_{\mathrm{SO}(n)}\ud \vph=  \prod_{i=1}^{n}\mathrm{Area}\left(\mathbb{S}^{i-1}\right)\, , \quad \mathrm{Area}\left(\mathbb{S}^{i-1}\right) =  2\frac{\pi^{i/2}}{\Gamma(i/2)}\,. 
\end{equation}

Finally, note that the spectral decomposition of $\rho_{\vet,\vph }$ is easily deduced from \eqref{nrho} by using the orthonormal set of column vectors of the rotation matrix $\mathcal{R}(\vph )$:
\begin{equation}
\label{specdecnrho}
\rho_{\vet,\vph } = \sum_{i=1}^n \left(\frac{1}{n} + \eta_i\right)\,|\vv_i\rg\lg \vv_i|\, , \quad \mathcal{R}(\vph )= (\vv_1, \dotsc, \vv_n)\, .
\end{equation}

\subsection{SO$(n)$ covariant integral quantisation}
 
Let us choose a set of $(n+2)(n-1)/2$ parameters $(\vet,\vph _0)$.   The covariant  integral quantisation of functions (or distributions) on the SO$(n)$ manifold is based on the resolution of the identity provided by the rotational transport of the matrix density $\rho_{\vet,\vph _0}$, similarly as what we have seen in the two-dimensional case \eqref{margomegamain}. The resolution of identity reads
\begin{equation}
\label{nresunit}
\bu= \int_{\mathrm{SO}(n)}\frac{\ud \vph }{c_n}\, \mathcal{R}(\vph )\, \rho_{\vet,\vph _0}\, {}^t\mathcal{R}(\vph )=\int_{\mathrm{SO}(n)}\frac{\ud \vph }{c_n}\,  \rho_{\vet,\vph \circ\vph _0}\,, 
\end{equation}
where $\ud \vph $ is the Haar measure on SO$(n)$ and $\vph \circ\vph _0$ stands for the group composition law in the parameter space. Equality \eqref{nresunit} is validated by Schur's Lemma. Indeed, due to the invariance of the Haar measure,  the right-hand side commutes with the (fundamental)  $n$-dimensional unitary irreducible representation $\vbe\mapsto \mathcal{R}(\vbe)$ of SO$(n)$. The constant $c_n$ is calculated by taking the trace on both sides,
\begin{equation}
\label{nconstc}
c_n= \frac{\mathrm{Vol} \left(\mathrm{SO}(n)\right)}{n}\,.
\end{equation}
Note that the resolution of the identity could be also proved by hand using the orthonormality between matrix elements of the fundamental (symmetric) UIR of SO$(n)$, $\vph\mapsto \mathcal{R}(\vph )$,
\begin{equation}
\label{northmatel}
\int_{\mathrm{SO}(n)}\frac{\ud \vph }{c_n}\, \mathcal{R}_{ij}(\vph ) \,\mathcal{R}_{i^{\prime}j^{\prime}}(\vph )= \delta_{ii^{\prime}}\,\delta_{jj^{\prime}}\,,
\end{equation}
 which, with $\sum_{i=1}^n \eta_i=0$, entails
 \begin{equation}
\label{northmatel1}
\int_{\mathrm{SO}(n)}\frac{\ud \vph }{c_n}\, \mathcal{R}(\vph ) \,\sfD(\vet)\,{}^t\mathcal{R}(\vph )= 0_{n\times n}\,.
\end{equation}

From \eqref{nresunit}, we can derive the quantisation of a real-valued function (or distribution) $f(\vph )$ on SO$(n)$. It  is precisely given by the map
\begin{equation}
\label{nquantmap}\begin{split}
f(\vph ) \mapsto A_f &=\int_{\mathrm{SO}(n)}\frac{\ud \vph }{c_n}\, f(\vph )\, \rho_{\vet,\vph \circ\vph _0}=\int_{\mathrm{SO}(n)}\frac{\ud \vph }{c_n}\, f\left(\vph \circ(\vph _0)^{-1}\right)\, \rho_{\vet,\vph }\\
&= \lg f \rg + \int_{\mathrm{SO}(n)}\frac{\ud \vph }{c_n}\, f\left(\vph \circ(\vph _0)^{-1}\right)\,\mathcal{R}(\vph ) \,\sfD(\vet)\,{}^t\mathcal{R}(\vph )\\
&\equiv \lg f \rg + \sfN_f\,. 
\end{split}
\end{equation}
In the above equation, $\lg f \rg$ denotes the mean value of $f$ on $\mathrm{SO}(n)$, 
\begin{equation}
\label{nmeanval}
\lg f \rg= \int_{\mathrm{SO}(n)}\frac{\ud \vph }{\mathrm{Vol} \left(\mathrm{SO}(n)\right)}\, f(\vph)\,.
\end{equation} 

We find the elements of the matrix $\sfN_f$ using the decomposition of the tensor product 
\begin{equation}
\label{ntensprod}
\mathcal{R} \otimes \mathcal{R}= \oplus_{\vt} \mathrm{T}^{\vt} \;,
\end{equation}
as a (finite) direct sum of UIR's $\mathrm{T}^{\vt}$ of SO$(n)$ ($\vt$ is a multi-index). Actually, the range of this multi-index is quite narrow in our case.  From the Peter-Weyl theorem for continuous functions on a compact group \cite{peter1927vollstandigkeit}, this expansion reads on the level of matrix elements of $\mathcal{R}$:
 \begin{equation}
\label{nclebsch}
\mathcal{R}_{il}(\vph)\,\mathcal{R}_{jl} (\vph)= \sum_{\vt,\vI,\vJ} C^{\vt}_{il,jl,\vI,\vJ}\mathrm{T}^{\vt}_{\vI,\vJ}(\vph)\, , 
\end{equation}
where the $C^{\vt}_{il,jl,\vI,\vJ}$ are Clebsch-Gordan coefficients and $\vI$, $\vJ$, are multi-indices.  The range of $\vet$ is obviously quite limited. 
On the other hand, with minimal assumptions on $f$, \textit{e.g.} square integrability on SO$(n)$, and from the completeness of the $\mathrm{T}^{\vt}_{\vI,\vJ}(\vph)$'s, one can Fourier expand $f$ as
\begin{equation}
\label{nfexp}
f(\vph)= \sum_{\vt,\vI,\vJ} f^{\vt}_{\vI,\vJ}\mathrm{T}^{\vt}_{\vI,\vJ}(\vph)\, ,
\end{equation}
and eventually,
\begin{equation}
\label{nfexpal0}
f\left(\vph \circ(\vph _0)^{-1}\right)= \sum_{\vt,\vI,\vJ,\vK} f^{\vt}_{\vI,\vJ}\mathrm{T}^{\vt}_{\vI,\vK}(\vph)\,\mathrm{T}^{\vt}_{\vJ,\vK}\left(\vph_0\right)\, ,
\end{equation}
With \eqref{nclebsch} and \eqref{nfexpal0} in hand, and from the orthogonality of the $\mathrm{T}^{\vt}_{\vI,\vJ}(\vph)$'s, the matrix elements of $\sfN_f$ read as
\begin{equation}
\label{Nfij}
\left[\sfN_f\right]_{ij} = \sum_{l=1}^n \eta_l\sum_{\vt,\vI,\vJ,\vK} f^{\vt}_{\vI,\vJ}\,C^{\vt}_{il,jl,\vI,\vJ}\,\mathrm{T}^{\vt}_{\vJ,\vK}\left(\vph_0\right)\,. 
\end{equation}

The covariance of the quantisation map \eqref{nquantmap} means
\begin{equation}
\label{nqcov}
\mathcal{R}(\vbe)\, A_f\, {}^t\mathcal{R}(\vbe)= A_{\mathrm{R}(\vbe)f}\, ,  \quad (\mathrm{R}(\vbe)f)(\vph ):=f\left((\vbe)^{-1}\circ\vph \right)\,. 
\end{equation}
Let us determine the function $f= 1/n + g$ such that $A_f$ is the density matrix $\rho_{\vze,\vbe}$, \textit{i.e.},
\begin{equation}
\label{rhorhon}
\rho_{\vze,\vbe} = \int_{\mathrm{SO}(n)}\frac{\ud \vph }{c_n}\, \left(\frac{1}{n}+ g(\vph )\right)\, \rho_{\vet,\vph \circ\vph _0}\,.  
\end{equation}
Taking the trace on both sides implies that 
\begin{equation}
\label{gsimpl}
\int_{\mathrm{SO}(n)}\ud \vph\, g(\vph)= 0\, , \quad \mbox{\textit{i.e.}}, \quad \lg g\rg = 0\,.
\end{equation}
Using \eqref{nresunit}, \eqref{northmatel1}, and \eqref{gsimpl} makes \eqref{rhorhon} equivalent to
\begin{equation}
\label{rhorhon1}
\mathcal{R} (\vbe)\,\sfD(\vze)\,{}^t\mathcal{R} (\vbe)= \int_{\mathrm{SO}(n)}\frac{\ud \vph }{c_n}\, g(\vph )\, \mathcal{R}(\vph \circ\vph _0)\,\sfD(\vet)\,{}^t\mathcal{R} (\vph \circ\vph _0)\,.  
\end{equation}
Using the SO$(n)$-covariance,   one gets for each $1\leq l\leq n$
\begin{equation}
\label{rhorhon2}
\frac{\zeta_l}{\eta_l}= \sum_{i,j}\int_{\mathrm{SO}(n)}\frac{\ud \vph }{c_n}\, g\left((\vbe)^{-1} \circ\vph\right)\, \mathcal{R}_{il}(\vph\circ \vph_0)\,\mathcal{R}_{jl} (\vph\circ \vph_0)\,.  
\end{equation}
Using the reduction formula \eqref{nclebsch}, expanding $g(\vph)=\sum_{\vt,\vI,\vJ} g^{\vt}_{\vI,\vJ}\mathrm{T}^{\vt}_{\vI,\vJ}(\vph)$ as in \eqref{nfexp},  using the SO$(n)$ representation law and the orthogonality of the corresponding matrix elements  yields the linear system for the components $g^{\vt}_{\vI,\vJ}$:
\begin{equation}
\label{linsystg}
\frac{\zeta_l}{\eta_l}=  \sum_{i,j,\vt,\vI,\vJ,\vK,\vI^{\prime}} g^{\vt}_{\vI^{\prime},\vJ}\, C^{\vt}_{il,jl,\vI\vJ}\mathrm{T}^{\vt}_{\vI^{\prime},\vJ}(\vbe)\,\mathrm{T}^{\vt}_{\vK,\vJ}(\vph_0)\,.
\end{equation}
Let us impose in our choice of $g$ the relations
\begin{equation}
\label{Cgh}
\sum_{i,j} g^{\vt}_{\vI^{\prime},\vJ}\, C^{\vt}_{il,jl,\vI\vJ}= h^{\vt}_{l,\vI}\delta_{\vI^{\prime} \vJ}\,.
\end{equation}
We next choose   $\vph_0= \left(\vbe\right)^{-1}$. We finally end with the following values of $\eta_l$:
\begin{equation}
\label{etal}
\eta_l= \frac{\zeta_l}{\sum_{\vI,\vt} h^{\vt}_{l,\vI}}\, . 
\end{equation}

Finally, let us consider the $n^2$-dimensional  Lie algebra $\mathfrak{gl}(n,\R)$ and its Cartan decomposition into antisymmetric and symmetric matrices,
\begin{equation}
\label{cartanglR}
\mathfrak{gl}(n,\R)= \mathfrak{k} \oplus \mathfrak{p}\, , 
\end{equation}
it is easy to prove that the quantisation \eqref{nquantmap} maps real-valued functions into the $n(n+1)/2$-dimensional space $\mathfrak{p}$ of symmetric matrices, \textit{i.e.}, the quantum observables, while   $\mathfrak{k}\simeq \mathfrak{so}(n)$ is the $n(n-1)/2$-dimensional Lie algebra of SO$(n)$. In particular, any element $X$ in the  sub-algebra $\mathfrak{k}$ may play the role of a (pseudo)-Hamiltonian as a generator of the  evolution (\textit{i.e.} unitary) 
operator $e^{tX}$ in $\R^n$.

\printbibliography

\end{document}